\documentclass[11pt]{myifcolog}

\usepackage{amsthm} 
\usepackage{amsmath,amssymb}         
 
\usepackage{xspace}  
\usepackage{color}
\usepackage{epsfig}    
\graphicspath{{.}{./figures/}} 

\usepackage{tikzfig}
\usepackage{keycommand}

\newkeycommand{\boxmap}[style=small box][1]{\,\tikz{\node[style=\commandkey{style}] (x) {$#1$};\draw(0,-1.25)--(x)--(0,1.25);}\,}
\newkeycommand{\dboxmap}[style=dbox][1]{\,\tikz{\node[style=\commandkey{style}] (x) {$#1$};\draw[boldedge] (0,-1.25)--(x)--(0,1.25);}\,}

\newkeycommand{\wirelabel}[style=white label]{\,\tikz{\node[style=\commandkey{style}]{};}\,\xspace}

\newkeycommand{\pointmap}[style=point][1]{\,\tikz{\node[style=\commandkey{style}] (x) at (0,-0.3) {$#1$};\node [style=none] (3) at (0, 0.9) {};\node [style=none] (3) at (0, -0.9) {};\draw (x)--(0,0.7);}\,}

\newkeycommand{\point}[style=point,estyle=][1]{\,\tikz{\node[style=\commandkey{style}] (x) at (0,-0.05) {$#1$};\draw [\commandkey{estyle}] (x)--(0,1.05);}\,}

\newkeycommand{\copointmap}[style=copoint][1]{\,\tikz{\node[style=\commandkey{style}] (x) at (0,0.3) {$#1$};\draw (x)--(0,-0.7);}\,}

\newkeycommand{\dpoint}[style=point][1]{\,\tikz{\node[style=\commandkey{style},doubled] (x) at (0,-0.3) {$#1$};\draw[boldedge] (x)--(0,0.7);}\,}

\newkeycommand{\kpoint}[style=kpoint,estyle=][1]{\,\tikz{\node[style=\commandkey{style}] (x) at (0,-0.05) {$#1$};\draw [\commandkey{estyle}] (x)--(0,1.05);}\,}

\newkeycommand{\kpointgray}[style=gray kpoint,estyle=][1]{\,\tikz{\node[style=\commandkey{style}] (x) at (0,-0.05) {$#1$};\draw [\commandkey{estyle}] (x)--(0,1.05);}\,}

\newkeycommand{\typedkpoint}[style=kpoint,edgestyle=][2]{%
\,\begin{tikzpicture}
  \begin{pgfonlayer}{nodelayer}
    \node [style=none] (0) at (0, 1) {};
    \node [style=\commandkey{style}] (1) at (0, -0.75) {$#2$};
    \node [style=right label] (2) at (0.25, 0.5) {$#1$};
  \end{pgfonlayer}
  \begin{pgfonlayer}{edgelayer}
    \draw [\commandkey{edgestyle}] (1) to (0.center);
  \end{pgfonlayer}
\end{tikzpicture}\,}

\newkeycommand{\typedkpointdag}[style=kpoint adjoint,edgestyle=][2]{%
\,\begin{tikzpicture}
  \begin{pgfonlayer}{nodelayer}
    \node [style=none] (0) at (0, -1) {};
    \node [style=\commandkey{style}] (1) at (0, 0.75) {$#2$};
    \node [style=right label] (2) at (0.25, -0.5) {$#1$};
  \end{pgfonlayer}
  \begin{pgfonlayer}{edgelayer}
    \draw [\commandkey{edgestyle}] (0.center) to (1);
  \end{pgfonlayer}
\end{tikzpicture}\,}

\newkeycommand{\bistate}[style=kpoint][1]{\,\begin{tikzpicture}
  \begin{pgfonlayer}{nodelayer}
    \node [style=\commandkey{style}, minimum width=1 cm, inner sep=2pt] (0) at (0, -0.25) {$#1$};
    \node [style=none] (1) at (-0.75, 0) {};
    \node [style=none] (2) at (0.75, 0) {};
    \node [style=none] (3) at (-0.75, 0.88) {};
    \node [style=none] (4) at (0.75, 0.88) {};
  \end{pgfonlayer}
  \begin{pgfonlayer}{edgelayer}
    \draw (1.center) to (3.center);
    \draw (2.center) to (4.center);
  \end{pgfonlayer}
\end{tikzpicture}\,}

\newkeycommand{\bistatebig}[style=kpoint][1]{\,\begin{tikzpicture}
  \begin{pgfonlayer}{nodelayer}
    \node [style=\commandkey{style}, minimum width=, minimum width=1.26 cm, inner sep=2pt] (0) at (0, -0.25) {$#1$};
    \node [style=none] (1) at (-1, 0) {};
    \node [style=none] (2) at (1, 0) {};
    \node [style=none] (3) at (-1, 0.88) {};
    \node [style=none] (4) at (1, 0.88) {};
  \end{pgfonlayer}
  \begin{pgfonlayer}{edgelayer}
    \draw (1.center) to (3.center);
    \draw (2.center) to (4.center);
  \end{pgfonlayer}
\end{tikzpicture}\,}

\newkeycommand{\dbistate}[style=dkpoint][1]{\,\begin{tikzpicture}
  \begin{pgfonlayer}{nodelayer}
    \node [style=\commandkey{style}, minimum width=1 cm, inner sep=2pt] (0) at (0, -0.25) {$#1$};
    \node [style=none] (1) at (-0.75, 0) {};
    \node [style=none] (2) at (0.75, 0) {};
    \node [style=none] (3) at (-0.75, 1.0) {};
    \node [style=none] (4) at (0.75, 1.0) {};
  \end{pgfonlayer}
  \begin{pgfonlayer}{edgelayer}
    \draw [boldedge] (1.center) to (3.center);
    \draw [boldedge] (2.center) to (4.center);
  \end{pgfonlayer}
\end{tikzpicture}\,}

\newkeycommand{\bistatebraket}[style1=kpoint,style2=kpointadj][2]{\,\begin{tikzpicture}
  \begin{pgfonlayer}{nodelayer}
    \node [style=\commandkey{style1}, minimum width=1 cm, inner sep=2pt] (0) at (0, -0.75) {$#1$};
    \node [style=none] (1) at (-0.75, -0.5) {};
    \node [style=none] (2) at (0.75, -0.5) {};
    \node [style=none] (3) at (-0.75, 0.5) {};
    \node [style=none] (4) at (0.75, 0.5) {};
    \node [style=\commandkey{style2}, minimum width=1 cm, inner sep=2pt] (5) at (0, 0.75) {$#2$};
  \end{pgfonlayer}
  \begin{pgfonlayer}{edgelayer}
    \draw (1.center) to (3.center);
    \draw (2.center) to (4.center);
  \end{pgfonlayer}
\end{tikzpicture}\,}

\newkeycommand{\weight}[style=dkpoint][1]{\,\begin{tikzpicture}
  \begin{pgfonlayer}{nodelayer}
    \node [style=\commandkey{style}] (0) at (0, -0.5) {$#1$};
    \node [style=upground] (1) at (0, 0.75) {};
  \end{pgfonlayer}
  \begin{pgfonlayer}{edgelayer}
    \draw [style=boldedge] (0) to (1);
  \end{pgfonlayer}
\end{tikzpicture}\,}

\newkeycommand{\dkpoint}[style=dkpoint][1]{\,\tikz{\node[style=\commandkey{style}] (x) at (0,-0.1) {$#1$};\draw[boldedge] (x)--(0,1);}\,}
\newkeycommand{\dkpointadj}[style=dkpointadj][1]{\,\tikz{\node[style=\commandkey{style}] (x) at (0,0.1) {$#1$};\draw[boldedge] (x)--(0,-1);}\,}
\newkeycommand{\dkpointtrans}[style=dkpointtrans][1]{\,\tikz{\node[style=\commandkey{style}] (x) at (0,0.1) {$#1$};\draw[boldedge] (x)--(0,-1);}\,}
\newkeycommand{\dkpointconj}[style=dkpointconj][1]{\,\tikz{\node[style=\commandkey{style}] (x) at (0,-0.1) {$#1$};\draw[boldedge] (x)--(0,1);}\,}

\newkeycommand{\blackkpoint}[style=black kpoint][1]{\,\tikz{\node[style=\commandkey{style}] (x) at (0,-0.1) {$#1$};\draw (x)--(0,1);}\,}
\newkeycommand{\blackkpointadj}[style=black kpointadj][1]{\,\tikz{\node[style=\commandkey{style}] (x) at (0,0.1) {$#1$};\draw (x)--(0,-1);}\,}
\newkeycommand{\blackdkpoint}[style=black dkpoint][1]{\,\tikz{\node[style=\commandkey{style}] (x) at (0,-0.1) {$#1$};\draw[boldedge] (x)--(0,1);}\,}
\newkeycommand{\blackdkpointadj}[style=black dkpointadj][1]{\,\tikz{\node[style=\commandkey{style}] (x) at (0,0.1) {$#1$};\draw[boldedge] (x)--(0,-1);}\,}

\newcommand{\trace}{\,\begin{tikzpicture}
  \begin{pgfonlayer}{nodelayer}
    \node [style=none] (0) at (0, -0.60) {};
    \node [style=upground] (1) at (0, 0.40) {};
  \end{pgfonlayer}
  \begin{pgfonlayer}{edgelayer}
    \draw [style=none] (0.center) to (1);
  \end{pgfonlayer}
\end{tikzpicture}\,}

\newcommand\discard\trace

\newkeycommand{\pointketbra}[style1=copoint,style2=point,estyle=][2]{\,%
\begin{tikzpicture}
  \begin{pgfonlayer}{nodelayer}
    \node [style=none] (0) at (0, -1.75) {};
    \node [style=\commandkey{style1}] (1) at (0, -0.75) {$#1$};
    \node [style=\commandkey{style2}] (2) at (0, 0.75) {$#2$};
    \node [style=none] (3) at (0, 1.75) {};
  \end{pgfonlayer}
  \begin{pgfonlayer}{edgelayer}
    \draw [\commandkey{estyle}] (0.center) to (1);
    \draw [\commandkey{estyle}] (2) to (3.center);
  \end{pgfonlayer}
\end{tikzpicture}\,}

\newkeycommand{\twocopointketbra}[style1=copoint,style2=copoint,style3=point][3]{\,%
\begin{tikzpicture}
  \begin{pgfonlayer}{nodelayer}
    \node [style=none] (0) at (-0.75, -1.75) {};
    \node [style=none] (1) at (0.75, -1.75) {};
    \node [style=\commandkey{style1}] (2) at (-0.75, -0.75) {$#1$};
    \node [style=\commandkey{style2}] (3) at (0.75, -0.75) {$#2$};
    \node [style=\commandkey{style3}] (4) at (0, 0.75) {$#3$};
    \node [style=none] (5) at (0, 1.75) {};
  \end{pgfonlayer}
  \begin{pgfonlayer}{edgelayer}
    \draw (0.center) to (2);
    \draw (1.center) to (3);
    \draw (4) to (5.center);
  \end{pgfonlayer}
\end{tikzpicture}\,}

\newkeycommand{\twopointketbra}[style1=copoint,style2=point,style3=point][3]{\,%
\begin{tikzpicture}
  \begin{pgfonlayer}{nodelayer}
    \node [style=none] (0) at (0, -1.75) {};
    \node [style=none] (1) at (-0.75, 1.75) {};
    \node [style=\commandkey{style1}] (2) at (0, -0.75) {$#1$};
    \node [style=\commandkey{style2}] (3) at (-0.75, 0.75) {$#2$};
    \node [style=\commandkey{style3}] (4) at (0.75, 0.75) {$#3$};
    \node [style=none] (5) at (0.75, 1.75) {};
  \end{pgfonlayer}
  \begin{pgfonlayer}{edgelayer}
    \draw (0.center) to (2);
    \draw (1.center) to (3);
    \draw (4) to (5.center);
  \end{pgfonlayer}
\end{tikzpicture}\,}

\newkeycommand{\pointbraket}[style1=point,style2=copoint,estyle=][2]{\,%
\begin{tikzpicture}
  \begin{pgfonlayer}{nodelayer}
    \node [style=\commandkey{style1}] (0) at (0, -0.625) {$#1$};
    \node [style=\commandkey{style2}] (1) at (0, 0.625) {$#2$};
  \end{pgfonlayer}
  \begin{pgfonlayer}{edgelayer}
    \draw [\commandkey{estyle}] (0) to (1);
  \end{pgfonlayer}
\end{tikzpicture}\,}

\newkeycommand{\kpointketbra}[style1=kpointdag,style2=kpoint,estyle=][2]{\,%
\begin{tikzpicture}
  \begin{pgfonlayer}{nodelayer}
    \node [style=none] (0) at (0, -1.9) {};
    \node [style=\commandkey{style1}] (1) at (0, -0.95) {$#1$};
    \node [style=\commandkey{style2}] (2) at (0, 0.95) {$#2$};
    \node [style=none] (3) at (0, 1.9) {};
  \end{pgfonlayer}
  \begin{pgfonlayer}{edgelayer}
    \draw [\commandkey{estyle}] (0.center) to (1);
    \draw [\commandkey{estyle}] (2) to (3.center);
  \end{pgfonlayer}
\end{tikzpicture}\,}

\newkeycommand{\kpointmap}[style1=kpoint,style2=map][2]{\,%
\begin{tikzpicture}
  \begin{pgfonlayer}{nodelayer}
    \node [style=\commandkey{style1}] (0) at (0, -1.1) {$#1$};
    \node [style=\commandkey{style2}] (1) at (0, 0.5) {$#2$};
    \node [style=none] (2) at (0, 1.75) {};
  \end{pgfonlayer}
  \begin{pgfonlayer}{edgelayer}
    \draw (0) to (1);
    \draw (1) to (2.center);
  \end{pgfonlayer}
\end{tikzpicture}%
\,}

\newkeycommand{\kpointmaptrans}[style1=kpointconj,style2=mapadj][2]{\,%
\begin{tikzpicture}
  \begin{pgfonlayer}{nodelayer}
    \node [style=\commandkey{style1}] (0) at (0, -1.1) {$#1$};
    \node [style=\commandkey{style2}] (1) at (0, 0.5) {$#2$};
    \node [style=none] (2) at (0, 1.75) {};
  \end{pgfonlayer}
  \begin{pgfonlayer}{edgelayer}
    \draw (0) to (1);
    \draw (1) to (2.center);
  \end{pgfonlayer}
\end{tikzpicture}%
\,}

\newkeycommand{\kpointmapadj}[style1=kpointadj,style2=map][2]{\,%
\begin{tikzpicture}
  \begin{pgfonlayer}{nodelayer}
    \node [style=\commandkey{style1}] (0) at (0, 1.1) {$#1$};
    \node [style=\commandkey{style2}] (1) at (0, -0.5) {$#2$};
    \node [style=none] (2) at (0, -1.75) {};
  \end{pgfonlayer}
  \begin{pgfonlayer}{edgelayer}
    \draw (0) to (1);
    \draw (1) to (2.center);
  \end{pgfonlayer}
\end{tikzpicture}%
\,}

\newkeycommand{\dkpointmap}[style1=dkpoint,style2=dmap][2]{\,%
\begin{tikzpicture}
  \begin{pgfonlayer}{nodelayer}
    \node [style=\commandkey{style1}] (0) at (0, -1.2) {$#1$};
    \node [style=\commandkey{style2}] (1) at (0, 0.4) {$#2$};
    \node [style=none] (2) at (0, 1.7) {};
  \end{pgfonlayer}
  \begin{pgfonlayer}{edgelayer}
    \draw[style=boldedge] (0) to (1);
    \draw[style=boldedge] (1) to (2.center);
  \end{pgfonlayer}
\end{tikzpicture}%
\,}

\newkeycommand{\dkpointmapx}[style1=dkpoint,style2=dmap][2]{\,%
\begin{tikzpicture}
  \begin{pgfonlayer}{nodelayer}
    \node [style=\commandkey{style1}] (0) at (0, -1.4) {$#1$};
    \node [style=\commandkey{style2}] (1) at (0, 0.6) {$#2$};
    \node [style=none] (2) at (0, 1.9) {};
  \end{pgfonlayer}
  \begin{pgfonlayer}{edgelayer}
    \draw[style=boldedge] (0) to (1);
    \draw[style=boldedge] (1) to (2.center);
  \end{pgfonlayer}
\end{tikzpicture}%
\,}

\newkeycommand{\boxcopointmapdag}[style1=map,style2=copoint][2]{\,%
\begin{tikzpicture}
  \begin{pgfonlayer}{nodelayer}
    \node [style=none] (0) at (0, -1.75) {};
    \node [style=\commandkey{style1}] (1) at (0, -0.5) {$#1$};
    \node [style=\commandkey{style2}] (2) at (0, 1) {$#2$};
  \end{pgfonlayer}
  \begin{pgfonlayer}{edgelayer}
    \draw (0.center) to (1);
    \draw (1) to (2);
  \end{pgfonlayer}
\end{tikzpicture}%
\,}

\newkeycommand{\kpointdagmap}[style1=map,style2=kpointdag][2]{\,%
\begin{tikzpicture}
  \begin{pgfonlayer}{nodelayer}
    \node [style=none] (0) at (0, -1.75) {};
    \node [style=\commandkey{style1}] (1) at (0, -0.5) {$#1$};
    \node [style=\commandkey{style2}] (2) at (0, 1.1) {$#2$};
  \end{pgfonlayer}
  \begin{pgfonlayer}{edgelayer}
    \draw (0.center) to (1);
    \draw (1) to (2);
  \end{pgfonlayer}
\end{tikzpicture}%
\,}

\newkeycommand{\sandwichmap}[style1=point,style2=map,style3=copoint][3]{\,%
\begin{tikzpicture}
  \begin{pgfonlayer}{nodelayer}
    \node [style=\commandkey{style1}] (0) at (0, -1.50) {$#1$};
    \node [style=\commandkey{style2}] (1) at (0, 0) {$#2$};
    \node [style=\commandkey{style3}] (2) at (0, 1.50) {$#3$};
  \end{pgfonlayer}
  \begin{pgfonlayer}{edgelayer}
    \draw (0) to (1);
    \draw (1) to (2);
  \end{pgfonlayer}
\end{tikzpicture}%
}

\newkeycommand{\kpointsandwichmap}[style1=kpoint,style2=map,style3=kpointdag][3]{\,%
\begin{tikzpicture}
  \begin{pgfonlayer}{nodelayer}
    \node [style=\commandkey{style1}] (0) at (0, -1.5) {$#1$};
    \node [style=\commandkey{style2}] (1) at (0, 0) {$#2$};
    \node [style=\commandkey{style3}] (2) at (0, 1.5) {$#3$};
  \end{pgfonlayer}
  \begin{pgfonlayer}{edgelayer}
    \draw (0) to (1);
    \draw (1) to (2);
  \end{pgfonlayer}
\end{tikzpicture}%
}

\newkeycommand{\sandwichmapconj}[style1=point,style2=mapconj,style3=copoint][3]{\,%
\begin{tikzpicture}
  \begin{pgfonlayer}{nodelayer}
    \node [style=\commandkey{style1}] (0) at (0, -1.50) {$#1$};
    \node [style=\commandkey{style2}] (1) at (0, 0) {$#2$};
    \node [style=\commandkey{style3}] (2) at (0, 1.50) {$#3$};
  \end{pgfonlayer}
  \begin{pgfonlayer}{edgelayer}
    \draw (0) to (1);
    \draw (1) to (2);
  \end{pgfonlayer}
\end{tikzpicture}%
}

\newkeycommand{\sandwichtwo}[style1=point,style2=map,style3=map,style4=copoint][4]{\,%
\begin{tikzpicture}
  \begin{pgfonlayer}{nodelayer}
    \node [style=\commandkey{style2}] (0) at (0, -1) {$#2$};
    \node [style=\commandkey{style3}] (1) at (0, 1) {$#3$};
    \node [style=\commandkey{style1}] (2) at (0, -2.6) {$#1$};
    \node [style=\commandkey{style4}] (3) at (0, 2.6) {$#4$};
  \end{pgfonlayer}
  \begin{pgfonlayer}{edgelayer}
    \draw (2) to (0);
    \draw (0) to (1);
    \draw (1) to (3);
  \end{pgfonlayer}
\end{tikzpicture}\,}

\newkeycommand{\longbraket}[style1=point,style2=copoint][2]{\,%
\begin{tikzpicture}
  \begin{pgfonlayer}{nodelayer}
    \node [style=\commandkey{style1}] (0) at (0, -2.6) {$#1$};
    \node [style=\commandkey{style2}] (1) at (0, 2.6) {$#2$};
  \end{pgfonlayer}
  \begin{pgfonlayer}{edgelayer}
    \draw (0) to (1);
  \end{pgfonlayer}
\end{tikzpicture}\,}

\newkeycommand{\onb}[style=point]{\ensuremath{\{\,\tikz{\node[style=\commandkey{style}] (x) at (0,-0.3) {$j$};\draw (x)--(0,0.7);}\,\}}\xspace}

\newkeycommand{\phase}[style=white phase dot][1]{\,\begin{tikzpicture}
    \begin{pgfonlayer}{nodelayer}
        \node [style=none] (0) at (0, 1) {};
        \node [style=\commandkey{style}] (2) at (0, -0) {$#1$}; 
        \node [style=none] (3) at (0, -1) {};
    \end{pgfonlayer}
    \begin{pgfonlayer}{edgelayer}
        \draw (2) to (0.center);
        \draw (3.center) to (2);
    \end{pgfonlayer}
\end{tikzpicture}\,}

\newkeycommand{\dphase}[style=white phase ddot][1]{\,\begin{tikzpicture}
    \begin{pgfonlayer}{nodelayer}
        \node [style=none] (0) at (0, 1.25) {};
        \node [style=\commandkey{style}] (2) at (0, -0) {$#1$};
        \node [style=none] (3) at (0, -1.25) {};
    \end{pgfonlayer}
    \begin{pgfonlayer}{edgelayer}
        \draw [boldedge] (2) to (0.center);
        \draw [boldedge] (3.center) to (2);
    \end{pgfonlayer}
\end{tikzpicture}\,}

\newkeycommand{\dphasepoint}[style=white phase ddot][1]{\,\begin{tikzpicture}[yshift=-3mm]
  \begin{pgfonlayer}{nodelayer}
    \node [style=none] (0) at (0, 1) {};
    \node [style=\commandkey{style}] (1) at (0, 0) {$#1$};
  \end{pgfonlayer}
  \begin{pgfonlayer}{edgelayer}
    \draw [style=boldedge] (0.center) to (1);
  \end{pgfonlayer}
\end{tikzpicture}\,}

\newkeycommand{\dphasepointgray}[style=gray phase ddot][1]{\,\begin{tikzpicture}[yshift=-3mm]
  \begin{pgfonlayer}{nodelayer}
    \node [style=none] (0) at (0, 1) {};
    \node [style=\commandkey{style}] (1) at (0, 0) {$#1$};
  \end{pgfonlayer}
  \begin{pgfonlayer}{edgelayer}
    \draw [style=boldedge] (0.center) to (1);
  \end{pgfonlayer}
\end{tikzpicture}\,}

\newkeycommand{\dphasecopoint}[style=white phase ddot][1]{\,\begin{tikzpicture}[yshift=3mm,yscale=-1]
  \begin{pgfonlayer}{nodelayer}
    \node [style=none] (0) at (0, 1) {};
    \node [style=\commandkey{style}] (1) at (0, 0) {$#1$};
  \end{pgfonlayer}
  \begin{pgfonlayer}{edgelayer}
    \draw [style=boldedge] (0.center) to (1);
  \end{pgfonlayer}
\end{tikzpicture}\,}

\newkeycommand{\dphasepointsm}[style=white phase ddot][1]{\,\begin{tikzpicture}[yshift=-3mm]
  \begin{pgfonlayer}{nodelayer}
    \node [style=none] (0) at (0, 1) {};
    \node [style=\commandkey{style}] (1) at (0, 0) {\footnotesize\!$#1$\!};
  \end{pgfonlayer}
  \begin{pgfonlayer}{edgelayer}
    \draw [style=boldedge] (0.center) to (1);
  \end{pgfonlayer}
\end{tikzpicture}\,}

\newkeycommand{\phasepoint}[style=white phase dot][1]{\,\begin{tikzpicture}[yshift=-3mm]
  \begin{pgfonlayer}{nodelayer}
    \node [style=none] (0) at (0, 1) {};
    \node [style=\commandkey{style}] (1) at (0, 0) {$#1$};
  \end{pgfonlayer}
  \begin{pgfonlayer}{edgelayer}
    \draw (0.center) to (1);
  \end{pgfonlayer}
\end{tikzpicture}\,}

\newkeycommand{\phasecopoint}[style=white phase dot][1]{\,\begin{tikzpicture}[yshift=3mm]
  \begin{pgfonlayer}{nodelayer}
    \node [style=none] (0) at (0, -1) {};
    \node [style=\commandkey{style}] (1) at (0, 0) {$#1$};
  \end{pgfonlayer}
  \begin{pgfonlayer}{edgelayer}
    \draw (0.center) to (1);
  \end{pgfonlayer}
\end{tikzpicture}\,}

\newkeycommand{\kpointbraket}[style1=kpoint,style2=kpoint adjoint,estyle=][2]{\,%
\begin{tikzpicture}
  \begin{pgfonlayer}{nodelayer}
    \node [style=\commandkey{style1}] (0) at (0, -0.735) {$#1$};
    \node [style=\commandkey{style2}] (1) at (0, 0.735) {$#2$};
  \end{pgfonlayer}
  \begin{pgfonlayer}{edgelayer}
    \draw [\commandkey{estyle}] (0) to (1);
  \end{pgfonlayer}
\end{tikzpicture}\,}

\newkeycommand{\kkpointbraket}[style1=kpoint,style2=kpoint adjoint,estyle=][2]{\,%
\begin{tikzpicture}
  \begin{pgfonlayer}{nodelayer}
    \node [style=\commandkey{style1}] (0) at (0, -0.685) {$#1$};
    \node [style=\commandkey{style2}] (1) at (0, 0.685) {$#2$};
  \end{pgfonlayer}
  \begin{pgfonlayer}{edgelayer}
    \draw [\commandkey{estyle}] (0) to (1);
  \end{pgfonlayer}
\end{tikzpicture}\,}

\newkeycommand{\kpointbraketx}[style1=kpoint,style2=kpoint adjoint,estyle=][2]{\,%
\begin{tikzpicture}
  \begin{pgfonlayer}{nodelayer}
    \node [style=\commandkey{style1}] (0) at (0, -0.885) {$#1$};
    \node [style=\commandkey{style2}] (1) at (0, 0.885) {$#2$};
  \end{pgfonlayer}
  \begin{pgfonlayer}{edgelayer}
    \draw [\commandkey{estyle}] (0) to (1);
  \end{pgfonlayer}
\end{tikzpicture}\,}

\tikzstyle{cdiag}=[matrix of math nodes, row sep=3em, column sep=3em, text height=1.5ex, text depth=0.25ex,inner sep=0.5em]
\tikzstyle{arrow above}=[transform canvas={yshift=0.5ex}]
\tikzstyle{arrow below}=[transform canvas={yshift=-0.5ex}]

\usepackage{tikz}
\usetikzlibrary{decorations.markings}
\usetikzlibrary{shapes.geometric}
\pgfdeclarelayer{edgelayer}
\pgfdeclarelayer{nodelayer}
\pgfsetlayers{edgelayer,nodelayer,main}

\tikzstyle{none}=[inner sep=0pt]
\definecolor{hexcolor0xff0000}{rgb}{1.000,0.000,0.000}
\definecolor{hexcolor0x000000}{rgb}{0.000,0.000,0.000}
\definecolor{hexcolor0x00ff00}{rgb}{0.000,1.000,0.000}
\definecolor{hexcolor0x000000}{rgb}{0.000,0.000,0.000}
\definecolor{hexcolor0xffff00}{rgb}{1.000,1.000,0.000}
\definecolor{hexcolor0xffffff}{rgb}{1.000,1.000,1.000}

\tikzstyle{rn}=[circle,fill=hexcolor0xff0000,draw=hexcolor0x000000,line width=0.8 pt]
\tikzstyle{gn}=[circle,fill=hexcolor0x00ff00,draw=hexcolor0x000000,line width=0.8 pt]
\tikzstyle{yn}=[circle,fill=hexcolor0xffff00,draw=hexcolor0x000000,line width=0.8 pt]
\tikzstyle{wn}=[circle,fill=hexcolor0xffffff,draw=hexcolor0x000000,line width=0.8 pt]
\tikzstyle{wnthick}=[circle,fill=hexcolor0xffffff,draw=hexcolor0x000000,line width=2.500]

\tikzstyle{simple}=[-,draw=hexcolor0x000000,line width=2.000]
\tikzstyle{arrow}=[-,draw=hexcolor0x000000,postaction={decorate},decoration={markings,mark=at position .5 with {\arrow{>}}},line width=2.000]
\tikzstyle{tick}=[-,draw=hexcolor0x000000,postaction={decorate},decoration={markings,mark=at position .5 with {\draw (0,-0.1) -- (0,0.1);}},line width=2.000]
\tikzstyle{halfthickness}=[-,draw=hexcolor0x000000,line width=0.500]
\tikzstyle{thick}=[-,draw=hexcolor0x000000,line width=2.500]
\tikzstyle{thicker}=[-,draw=hexcolor0x000000,line width=4.000]

\tikzstyle{env}=[copoint,regular polygon rotate=0,minimum width=0.2cm, fill=black]

\tikzstyle{probs}=[shape=semicircle,fill=white,draw=black,shape border rotate=180,minimum width=1.2cm]

\tikzstyle{every picture}=[baseline=-0.25em,scale=0.5]
\tikzstyle{dotpic}=[] 
\tikzstyle{diredges}=[every to/.style={diredge}]
\tikzstyle{math matrix}=[matrix of math nodes,left delimiter=(,right delimiter=),inner sep=2pt,column sep=1em,row sep=0.5em,nodes={inner sep=0pt},text height=1.5ex, text depth=0.25ex]

\tikzstyle{inline text}=[text height=1.5ex, text depth=0.25ex,yshift=0.5mm]
\tikzstyle{label}=[font=\footnotesize,text height=1.5ex, text depth=0.25ex,yshift=0.5mm]
\tikzstyle{left label}=[label,anchor=east,xshift=1.5mm]
\tikzstyle{right label}=[label,anchor=west,xshift=-1.5mm]

\tikzstyle{braceedge}=[decorate,decoration={brace,amplitude=2mm,raise=-1mm}]
\tikzstyle{small braceedge}=[decorate,decoration={brace,amplitude=1mm,raise=-1mm}]

\tikzstyle{doubled}=[line width=1.6pt] 
\tikzstyle{boldedge}=[doubled,shorten <=-0.17mm,shorten >=-0.17mm]
\tikzstyle{boldedgegray}=[doubled,gray,shorten <=-0.17mm,shorten >=-0.17mm]

\tikzstyle{semidoubled}=[line width=1.4pt] 
\tikzstyle{semiboldedgegray}=[semidoubled,gray,shorten <=-0.17mm,shorten >=-0.17mm]

\tikzstyle{boldedgedashed}=[very thick,dashed,shorten <=-0.17mm,shorten >=-0.17mm]
\tikzstyle{vboldedgedashed}=[doubled,dashed,shorten <=-0.17mm,shorten >=-0.17mm]
\tikzstyle{left hook arrow}=[left hook-latex]
\tikzstyle{right hook arrow}=[right hook-latex]
\tikzstyle{sembracket}=[line width=0.5pt,shorten <=-0.07mm,shorten >=-0.07mm]

\tikzstyle{causal edge}=[->,thick,gray]
\tikzstyle{causal nondir}=[thick,gray]
\tikzstyle{timeline}=[thick,gray, dashed]

\tikzstyle{cedge}=[<->,thick,gray!70!white]

\tikzstyle{empty diagram}=[draw=gray!40!white,dashed,shape=rectangle,minimum width=1cm,minimum height=1cm]
\tikzstyle{empty diagram small}=[draw=gray!50!white,dashed,shape=rectangle,minimum width=0.6cm,minimum height=0.5cm]

\tikzstyle{dot}=[inner sep=0mm,minimum width=2mm,minimum height=2mm,draw,shape=circle]
\tikzstyle{ddot}=[inner sep=0mm, doubled, minimum width=2.5mm,minimum height=2.5mm,draw,shape=circle]

\tikzstyle{black dot}=[dot,fill=black]
\tikzstyle{white dot}=[dot,fill=white,,text depth=-0.2mm]
\tikzstyle{green dot}=[white dot] 
\tikzstyle{gray dot}=[dot,fill=gray!40!white,,text depth=-0.2mm]

 \tikzstyle{red dot}=[dot,fill=red,font=\color{white}]

\tikzstyle{black ddot}=[ddot,fill=black]
\tikzstyle{white ddot}=[ddot,fill=white]
\tikzstyle{gray ddot}=[ddot,fill=gray!40!white]

\tikzstyle{gray edge}=[gray!40!white]

\tikzstyle{small dot}=[inner sep=0.5mm,minimum width=0pt,minimum height=0pt,draw,shape=circle]

\tikzstyle{small black dot}=[small dot,fill=black]
\tikzstyle{small white dot}=[small dot,fill=white]
\tikzstyle{small gray dot}=[small dot,fill=gray!40!white]

\tikzstyle{causal dot}=[inner sep=0.4mm,minimum width=0pt,minimum height=0pt,draw=white,shape=circle,fill=gray!40!white]

\tikzstyle{phase dimensions}=[minimum size=5mm,font=\footnotesize,rectangle,rounded corners=2.5mm,inner sep=0.2mm,outer sep=-2mm]
\tikzstyle{dphase dimensions}=[minimum size=5mm,font=\footnotesize,rectangle,rounded corners=2.5mm,inner sep=0.2mm,outer sep=-2mm]

\tikzstyle{white phase dot}=[dot,fill=white,phase dimensions]
\tikzstyle{white phase ddot}=[ddot,fill=white,dphase dimensions]

\tikzstyle{white rect ddot}=[draw=black,fill=white,doubled,minimum size=5mm,font=\footnotesize,rectangle,rounded corners=2.5mm,inner sep=0.2mm]
\tikzstyle{gray rect ddot}=[draw=black,fill=gray!40!white,doubled,minimum size=6mm,font=\footnotesize,rectangle,rounded corners=3mm]

\tikzstyle{gray phase dot}=[dot,fill=gray!40!white,phase dimensions]
\tikzstyle{gray phase ddot}=[ddot,fill=gray!40!white,dphase dimensions]
\tikzstyle{grey phase dot}=[gray phase dot]
\tikzstyle{grey phase ddot}=[gray phase ddot]

\tikzstyle{small phase dimensions}=[minimum size=4mm,font=\tiny,rectangle,rounded corners=2mm,inner sep=0.2mm,outer sep=-2mm]
\tikzstyle{small dphase dimensions}=[minimum size=4mm,font=\tiny,rectangle,rounded corners=2mm,inner sep=0.2mm,outer sep=-2mm]

\tikzstyle{small gray phase dot}=[dot,fill=gray!40!white,small phase dimensions]
\tikzstyle{small gray phase ddot}=[ddot,fill=gray!40!white,small dphase dimensions]

\tikzstyle{small map}=[draw,shape=rectangle,minimum height=4mm,minimum width=4mm,fill=white]

\tikzstyle{cnot}=[fill=white,shape=circle,inner sep=-1.4pt]

\tikzstyle{asym hadamard}=[fill=white,draw,shape=NEbox,inner sep=0.6mm,font=\footnotesize,minimum height=4mm]
\tikzstyle{asym hadamard conj}=[fill=white,draw,shape=NWbox,inner sep=0.6mm,font=\footnotesize,minimum height=4mm]
\tikzstyle{asym hadamard dag}=[fill=white,draw,shape=SEbox,inner sep=0.6mm,font=\footnotesize,minimum height=4mm]

\tikzstyle{hadamard}=[fill=white,draw,inner sep=0.6mm,font=\footnotesize,minimum height=4mm,minimum width=4mm]
\tikzstyle{small hadamard}=[fill=white,draw,inner sep=0.6mm,minimum height=1.5mm,minimum width=1.5mm]
\tikzstyle{dhadamard}=[hadamard,doubled]
\tikzstyle{small dhadamard}=[small hadamard,doubled]
\tikzstyle{small dhadamard rotate}=[small hadamard,doubled,rotate=45]
\tikzstyle{antipode}=[white dot,inner sep=0.3mm,font=\footnotesize]

\tikzstyle{scalar}=[diamond,draw,inner sep=0.5pt,font=\small]
\tikzstyle{dscalar}=[diamond,doubled, draw,inner sep=0.5pt,font=\small]

\tikzstyle{small box}=[rectangle,inline text,fill=white,draw,minimum height=5mm,yshift=-0.5mm,minimum width=5mm,font=\small]
\tikzstyle{small gray box}=[small box,fill=gray!30]
\tikzstyle{medium box}=[rectangle,inline text,fill=white,draw,minimum height=5mm,yshift=-0.5mm,minimum width=10mm,font=\small]
\tikzstyle{square box}=[small box] 
\tikzstyle{medium gray box}=[small box,fill=gray!30]
\tikzstyle{semilarge box}=[rectangle,inline text,fill=white,draw,minimum height=5mm,yshift=-0.5mm,minimum width=12.5mm,font=\small]
\tikzstyle{large box}=[rectangle,inline text,fill=white,draw,minimum height=5mm,yshift=-0.5mm,minimum width=15mm,font=\small]
\tikzstyle{large gray box}=[small box,fill=gray!30]

\tikzstyle{Bayes box}=[rectangle,fill=black,draw, minimum height=3mm, minimum width=3mm]

\tikzstyle{gray square point}=[small box,fill=gray!50]

\tikzstyle{dphase box white}=[dhadamard]
\tikzstyle{dphase box gray}=[dhadamard,fill=gray!50!white]

\tikzstyle{point}=[regular polygon,regular polygon sides=3,draw,scale=0.75,inner sep=-0.5pt,minimum width=9mm,fill=white,regular polygon rotate=180]
\tikzstyle{copoint}=[regular polygon,regular polygon sides=3,draw,scale=0.75,inner sep=-0.5pt,minimum width=9mm,fill=white]
\tikzstyle{dpoint}=[point,doubled]
\tikzstyle{dcopoint}=[copoint,doubled]

\tikzstyle{wide copoint}=[fill=white,draw,shape=isosceles triangle,shape border rotate=90,isosceles triangle stretches=true,inner sep=0pt,minimum width=1.5cm,minimum height=6.12mm]
\tikzstyle{wide point}=[fill=white,draw,shape=isosceles triangle,shape border rotate=-90,isosceles triangle stretches=true,inner sep=0pt,minimum width=1.5cm,minimum height=6.12mm,yshift=-0.0mm]
\tikzstyle{wide point plus}=[fill=white,draw,shape=isosceles triangle,shape border rotate=-90,isosceles triangle stretches=true,inner sep=0pt,minimum width=1.74cm,minimum height=7mm,yshift=-0.0mm]

\tikzstyle{wide dpoint}=[fill=white,doubled,draw,shape=isosceles triangle,shape border rotate=-90,isosceles triangle stretches=true,inner sep=0pt,minimum width=1.5cm,minimum height=6.12mm,yshift=-0.0mm]
\tikzstyle{wide dcopoint}=[fill=white,doubled,draw,shape=isosceles triangle,shape border rotate=90,isosceles triangle stretches=true,inner sep=0pt,minimum width=1.5cm,minimum height=6.12mm,yshift=-0.0mm]

\tikzstyle{tinypoint}=[regular polygon,regular polygon sides=3,draw,scale=0.55,inner sep=-0.15pt,minimum width=6mm,fill=white,regular polygon rotate=180]

\tikzstyle{white point}=[point]
\tikzstyle{white dpoint}=[dpoint]
\tikzstyle{green point}=[white point] 
\tikzstyle{white copoint}=[copoint]
\tikzstyle{gray point}=[point,fill=gray!40!white]
\tikzstyle{gray dpoint}=[gray point,doubled]
\tikzstyle{red point}=[gray point] 
\tikzstyle{gray copoint}=[copoint,fill=gray!40!white]
\tikzstyle{gray dcopoint}=[gray copoint,doubled]

\tikzstyle{white point guide}=[regular polygon,regular polygon sides=3,font=\scriptsize,draw,scale=0.65,inner sep=-0.5pt,minimum width=9mm,fill=white,regular polygon rotate=180]

\tikzstyle{black point}=[point,fill=black,font=\color{white}]
\tikzstyle{black copoint}=[copoint,fill=black,font=\color{white}]

\tikzstyle{tiny gray point}=[tinypoint,fill=gray!40!white]

\tikzstyle{diredge}=[->]
\tikzstyle{ddiredge}=[<->]
\tikzstyle{rdiredge}=[<-]
\tikzstyle{thickdiredge}=[->, very thick]
\tikzstyle{pointer edge}=[->,very thick,gray]
\tikzstyle{pointer edge part}=[very thick,gray]
\tikzstyle{dashed edge}=[dashed]
\tikzstyle{thick dashed edge}=[very thick,dashed]
\tikzstyle{thick gray dashed edge}=[thick dashed edge,gray!40]
\tikzstyle{thick map edge}=[very thick,|->]

\makeatletter
\newcommand{\boxshape}[3]{%
\pgfdeclareshape{#1}{
\inheritsavedanchors[from=rectangle] 
\inheritanchorborder[from=rectangle]
\inheritanchor[from=rectangle]{center}
\inheritanchor[from=rectangle]{north}
\inheritanchor[from=rectangle]{south}
\inheritanchor[from=rectangle]{west}
\inheritanchor[from=rectangle]{east}
\backgroundpath{
\southwest \pgf@xa=\pgf@x \pgf@ya=\pgf@y
\northeast \pgf@xb=\pgf@x \pgf@yb=\pgf@y

\@tempdima=#2
\@tempdimb=#3

\pgfpathmoveto{\pgfpoint{\pgf@xa - 5pt + \@tempdima}{\pgf@ya}}
\pgfpathlineto{\pgfpoint{\pgf@xa - 5pt - \@tempdima}{\pgf@yb}}
\pgfpathlineto{\pgfpoint{\pgf@xb + 5pt + \@tempdimb}{\pgf@yb}}
\pgfpathlineto{\pgfpoint{\pgf@xb + 5pt - \@tempdimb}{\pgf@ya}}
\pgfpathlineto{\pgfpoint{\pgf@xa - 5pt + \@tempdima}{\pgf@ya}}
\pgfpathclose
}
}}

\boxshape{NEbox}{0pt}{5pt}
\boxshape{SEbox}{0pt}{-5pt}
\boxshape{NWbox}{5pt}{0pt}
\boxshape{SWbox}{-5pt}{0pt}
\boxshape{EBox}{-3pt}{3pt}
\boxshape{WBox}{3pt}{-3pt}
\makeatother

\tikzstyle{cloud}=[shape=cloud,draw,minimum width=1.5cm,minimum height=1.5cm]

\tikzstyle{map}=[draw,shape=NEbox,inner sep=2pt,minimum height=6mm,fill=white]
\tikzstyle{dashedmap}=[draw,dashed,shape=NEbox,inner sep=2pt,minimum height=6mm,fill=white]
\tikzstyle{mapdag}=[draw,shape=SEbox,inner sep=2pt,minimum height=6mm,fill=white]
\tikzstyle{mapadj}=[draw,shape=SEbox,inner sep=2pt,minimum height=6mm,fill=white]
\tikzstyle{maptrans}=[draw,shape=SWbox,inner sep=2pt,minimum height=6mm,fill=white]
\tikzstyle{mapconj}=[draw,shape=NWbox,inner sep=2pt,minimum height=6mm,fill=white]

\tikzstyle{medium map}=[draw,shape=NEbox,inner sep=2pt,minimum height=6mm,fill=white,minimum width=7mm]
\tikzstyle{medium map dag}=[draw,shape=SEbox,inner sep=2pt,minimum height=6mm,fill=white,minimum width=7mm]
\tikzstyle{medium map adj}=[draw,shape=SEbox,inner sep=2pt,minimum height=6mm,fill=white,minimum width=7mm]
\tikzstyle{medium map trans}=[draw,shape=SWbox,inner sep=2pt,minimum height=6mm,fill=white,minimum width=7mm]
\tikzstyle{medium map conj}=[draw,shape=NWbox,inner sep=2pt,minimum height=6mm,fill=white,minimum width=7mm]
\tikzstyle{semilarge map}=[draw,shape=NEbox,inner sep=2pt,minimum height=6mm,fill=white,minimum width=9.5mm]
\tikzstyle{semilarge map trans}=[draw,shape=SWbox,inner sep=2pt,minimum height=6mm,fill=white,minimum width=9.5mm]
\tikzstyle{semilarge map adj}=[draw,shape=SEbox,inner sep=2pt,minimum height=6mm,fill=white,minimum width=9.5mm]
\tikzstyle{semilarge map dag}=[draw,shape=SEbox,inner sep=2pt,minimum height=6mm,fill=white,minimum width=9.5mm]
\tikzstyle{semilarge map conj}=[draw,shape=NWbox,inner sep=2pt,minimum height=6mm,fill=white,minimum width=9.5mm]
\tikzstyle{large map}=[draw,shape=NEbox,inner sep=2pt,minimum height=6mm,fill=white,minimum width=12mm]
\tikzstyle{large map conj}=[draw,shape=NWbox,inner sep=2pt,minimum height=6mm,fill=white,minimum width=12mm]
\tikzstyle{very large map}=[draw,shape=NEbox,inner sep=2pt,minimum height=6mm,fill=white,minimum width=17mm]

\tikzstyle{medium dmap}=[draw,doubled,shape=NEbox,inner sep=2pt,minimum height=6mm,fill=white,minimum width=7mm]
\tikzstyle{medium dmap dag}=[draw,doubled,shape=SEbox,inner sep=2pt,minimum height=6mm,fill=white,minimum width=7mm]
\tikzstyle{medium dmap adj}=[draw,doubled,shape=SEbox,inner sep=2pt,minimum height=6mm,fill=white,minimum width=7mm]
\tikzstyle{medium dmap trans}=[draw,doubled,shape=SWbox,inner sep=2pt,minimum height=6mm,fill=white,minimum width=7mm]
\tikzstyle{medium dmap conj}=[draw,doubled,shape=NWbox,inner sep=2pt,minimum height=6mm,fill=white,minimum width=7mm]
\tikzstyle{semilarge dmap}=[draw,doubled,shape=NEbox,inner sep=2pt,minimum height=6mm,fill=white,minimum width=9.5mm]
\tikzstyle{semilarge dmap trans}=[draw,doubled,shape=SWbox,inner sep=2pt,minimum height=6mm,fill=white,minimum width=9.5mm]
\tikzstyle{semilarge dmap adj}=[draw,doubled,shape=SEbox,inner sep=2pt,minimum height=6mm,fill=white,minimum width=9.5mm]
\tikzstyle{semilarge dmap dag}=[draw,doubled,shape=SEbox,inner sep=2pt,minimum height=6mm,fill=white,minimum width=9.5mm]
\tikzstyle{semilarge dmap conj}=[draw,doubled,shape=NWbox,inner sep=2pt,minimum height=6mm,fill=white,minimum width=9.5mm]
\tikzstyle{large dmap}=[draw,doubled,shape=NEbox,inner sep=2pt,minimum height=6mm,fill=white,minimum width=12mm]
\tikzstyle{large dmap conj}=[draw,doubled,shape=NWbox,inner sep=2pt,minimum height=6mm,fill=white,minimum width=12mm]
\tikzstyle{large dmap trans}=[draw,doubled,shape=SWbox,inner sep=2pt,minimum height=6mm,fill=white,minimum width=12mm]
\tikzstyle{large dmap adj}=[draw,doubled,shape=SEbox,inner sep=2pt,minimum height=6mm,fill=white,minimum width=12mm]
\tikzstyle{large dmap dag}=[draw,doubled,shape=SEbox,inner sep=2pt,minimum height=6mm,fill=white,minimum width=12mm]
\tikzstyle{very large dmap}=[draw,doubled,shape=NEbox,inner sep=2pt,minimum height=6mm,fill=white,minimum width=19.5mm]

\tikzstyle{muxbox}=[draw,shape=rectangle,minimum height=3mm,minimum width=3mm,fill=white]
\tikzstyle{dmuxbox}=[muxbox,doubled]

\tikzstyle{box}=[draw,shape=rectangle,inner sep=2pt,minimum height=6mm,minimum width=6mm,fill=white]
\tikzstyle{dbox}=[draw,doubled,shape=rectangle,inner sep=2pt,minimum height=6mm,minimum width=6mm,fill=white]
\tikzstyle{dmap}=[draw,doubled,shape=NEbox,inner sep=2pt,minimum height=6mm,fill=white]
\tikzstyle{dmapdag}=[draw,doubled,shape=SEbox,inner sep=2pt,minimum height=6mm,fill=white]
\tikzstyle{dmapadj}=[draw,doubled,shape=SEbox,inner sep=2pt,minimum height=6mm,fill=white]
\tikzstyle{dmaptrans}=[draw,doubled,shape=SWbox,inner sep=2pt,minimum height=6mm,fill=white]
\tikzstyle{dmapconj}=[draw,doubled,shape=NWbox,inner sep=2pt,minimum height=6mm,fill=white]

\tikzstyle{ddmap}=[draw,doubled,dashed,shape=NEbox,inner sep=2pt,minimum height=6mm,fill=white]
\tikzstyle{ddmapdag}=[draw,doubled,dashed,shape=SEbox,inner sep=2pt,minimum height=6mm,fill=white]
\tikzstyle{ddmapadj}=[draw,doubled,dashed,shape=SEbox,inner sep=2pt,minimum height=6mm,fill=white]
\tikzstyle{ddmaptrans}=[draw,doubled,dashed,shape=SWbox,inner sep=2pt,minimum height=6mm,fill=white]
\tikzstyle{ddmapconj}=[draw,doubled,dashed,shape=NWbox,inner sep=2pt,minimum height=6mm,fill=white]

\boxshape{sNEbox}{0pt}{3pt}
\boxshape{sSEbox}{0pt}{-3pt}
\boxshape{sNWbox}{3pt}{0pt}
\boxshape{sSWbox}{-3pt}{0pt}
\tikzstyle{smap}=[draw,shape=sNEbox,fill=white]
\tikzstyle{smapdag}=[draw,shape=sSEbox,fill=white]
\tikzstyle{smapadj}=[draw,shape=sSEbox,fill=white]
\tikzstyle{smaptrans}=[draw,shape=sSWbox,fill=white]
\tikzstyle{smapconj}=[draw,shape=sNWbox,fill=white]

\tikzstyle{dsmap}=[draw,dashed,shape=sNEbox,fill=white]
\tikzstyle{dsmapdag}=[draw,dashed,shape=sSEbox,fill=white]
\tikzstyle{dsmaptrans}=[draw,dashed,shape=sSWbox,fill=white]
\tikzstyle{dsmapconj}=[draw,dashed,shape=sNWbox,fill=white]

\boxshape{mNEbox}{0pt}{10pt}
\boxshape{mSEbox}{0pt}{-10pt}
\boxshape{mNWbox}{10pt}{0pt}
\boxshape{mSWbox}{-10pt}{0pt}
\tikzstyle{mmap}=[draw,shape=mNEbox]
\tikzstyle{mmapdag}=[draw,shape=mSEbox]
\tikzstyle{mmaptrans}=[draw,shape=mSWbox]
\tikzstyle{mmapconj}=[draw,shape=mNWbox]

\tikzstyle{mmapgray}=[draw,fill=gray!40!white,shape=mNEbox]
\tikzstyle{smapgray}=[draw,fill=gray!40!white,shape=sNEbox]

\makeatletter
\pgfdeclareshape{cornerpoint}{
\inheritsavedanchors[from=rectangle] 
\inheritanchorborder[from=rectangle]
\inheritanchor[from=rectangle]{center}
\inheritanchor[from=rectangle]{north}
\inheritanchor[from=rectangle]{south}
\inheritanchor[from=rectangle]{west}
\inheritanchor[from=rectangle]{east}
\backgroundpath{
\southwest \pgf@xa=\pgf@x \pgf@ya=\pgf@y
\northeast \pgf@xb=\pgf@x \pgf@yb=\pgf@y

\pgfmathsetmacro{\pgf@shorten@left}{\pgfkeysvalueof{/tikz/shorten left}}
\pgfmathsetmacro{\pgf@shorten@right}{\pgfkeysvalueof{/tikz/shorten right}}

\pgfpathmoveto{\pgfpoint{0.5 * (\pgf@xa + \pgf@xb)}{\pgf@ya - 5pt}}
\pgfpathlineto{\pgfpoint{\pgf@xa - 8pt + \pgf@shorten@left}{\pgf@yb - 1.5 * \pgf@shorten@left}}
\pgfpathlineto{\pgfpoint{\pgf@xa - 8pt + \pgf@shorten@left}{\pgf@yb}}
\pgfpathlineto{\pgfpoint{\pgf@xb + 8pt - \pgf@shorten@right}{\pgf@yb}}
\pgfpathlineto{\pgfpoint{\pgf@xb + 8pt - \pgf@shorten@right}{\pgf@yb - 1.5 * \pgf@shorten@right}}
\pgfpathclose
}
}

\pgfdeclareshape{cornercopoint}{
\inheritsavedanchors[from=rectangle] 
\inheritanchorborder[from=rectangle]
\inheritanchor[from=rectangle]{center}
\inheritanchor[from=rectangle]{north}
\inheritanchor[from=rectangle]{south}
\inheritanchor[from=rectangle]{west}
\inheritanchor[from=rectangle]{east}
\backgroundpath{
\southwest \pgf@xa=\pgf@x \pgf@ya=\pgf@y
\northeast \pgf@xb=\pgf@x \pgf@yb=\pgf@y

\pgfmathsetmacro{\pgf@shorten@left}{\pgfkeysvalueof{/tikz/shorten left}}
\pgfmathsetmacro{\pgf@shorten@right}{\pgfkeysvalueof{/tikz/shorten right}}

\pgfpathmoveto{\pgfpoint{0.5 * (\pgf@xa + \pgf@xb)}{\pgf@yb + 5pt}}
\pgfpathlineto{\pgfpoint{\pgf@xa - 8pt + \pgf@shorten@left}{\pgf@ya + 1.5 * \pgf@shorten@left}}
\pgfpathlineto{\pgfpoint{\pgf@xa - 8pt + \pgf@shorten@left}{\pgf@ya}}
\pgfpathlineto{\pgfpoint{\pgf@xb + 8pt - \pgf@shorten@right}{\pgf@ya}}
\pgfpathlineto{\pgfpoint{\pgf@xb + 8pt - \pgf@shorten@right}{\pgf@ya + 1.5 * \pgf@shorten@right}}
\pgfpathclose
}
}

\pgfdeclareshape{langpoint}{
\inheritsavedanchors[from=rectangle] 
\inheritanchorborder[from=rectangle]
\inheritanchor[from=rectangle]{center}
\inheritanchor[from=rectangle]{north}
\inheritanchor[from=rectangle]{south}
\inheritanchor[from=rectangle]{west}
\inheritanchor[from=rectangle]{east}
\backgroundpath{
\southwest \pgf@xa=\pgf@x \pgf@ya=\pgf@y
\northeast \pgf@xb=\pgf@x \pgf@yb=\pgf@y

\pgfmathsetmacro{\pgf@shorten@left}{\pgfkeysvalueof{/tikz/shorten left}}
\pgfmathsetmacro{\pgf@shorten@right}{\pgfkeysvalueof{/tikz/shorten right}}

\pgfpathmoveto{\pgfpoint{0.5 * (\pgf@xa + \pgf@xb)}{\pgf@ya - 2pt}}
\pgfpathlineto{\pgfpoint{\pgf@xa - 8pt}{\pgf@yb - 3 * \pgf@shorten@left + 5pt}} 
\pgfpathlineto{\pgfpoint{\pgf@xa - 8pt}{\pgf@yb -1pt}}
\pgfpathlineto{\pgfpoint{\pgf@xb + 8pt}{\pgf@yb -1pt}}
\pgfpathlineto{\pgfpoint{\pgf@xb + 8pt}{\pgf@yb - 3 * \pgf@shorten@left + 5pt}}
\pgfpathclose
}
}

\pgfdeclareshape{langcopointhigh}{
\inheritsavedanchors[from=rectangle] 
\inheritanchorborder[from=rectangle]
\inheritanchor[from=rectangle]{center}
\inheritanchor[from=rectangle]{north}
\inheritanchor[from=rectangle]{south}
\inheritanchor[from=rectangle]{west}
\inheritanchor[from=rectangle]{east}
\backgroundpath{
\southwest \pgf@xa=\pgf@x \pgf@ya=\pgf@y
\northeast \pgf@xb=\pgf@x \pgf@yb=\pgf@y

\pgfmathsetmacro{\pgf@shorten@left}{\pgfkeysvalueof{/tikz/shorten left}}
\pgfmathsetmacro{\pgf@shorten@right}{\pgfkeysvalueof{/tikz/shorten right}}

\pgfpathmoveto{\pgfpoint{0.5 * (\pgf@xa + \pgf@xb)}{\pgf@yb +4pt}}
\pgfpathlineto{\pgfpoint{\pgf@xa - 8pt}{\pgf@ya + 3 * \pgf@shorten@left -1pt}} 
\pgfpathlineto{\pgfpoint{\pgf@xa - 8pt}{\pgf@ya + 1pt}}
\pgfpathlineto{\pgfpoint{\pgf@xb + 8pt}{\pgf@ya + 1pt}}
\pgfpathlineto{\pgfpoint{\pgf@xb + 8pt}{\pgf@ya + 3 * \pgf@shorten@left -1pt}}
\pgfpathclose
}
} 

\pgfdeclareshape{langcopoint}{
\inheritsavedanchors[from=rectangle] 
\inheritanchorborder[from=rectangle]
\inheritanchor[from=rectangle]{center}
\inheritanchor[from=rectangle]{north}
\inheritanchor[from=rectangle]{south}
\inheritanchor[from=rectangle]{west}
\inheritanchor[from=rectangle]{east}
\backgroundpath{
\southwest \pgf@xa=\pgf@x \pgf@ya=\pgf@y
\northeast \pgf@xb=\pgf@x \pgf@yb=\pgf@y

\pgfmathsetmacro{\pgf@shorten@left}{\pgfkeysvalueof{/tikz/shorten left}}
\pgfmathsetmacro{\pgf@shorten@right}{\pgfkeysvalueof{/tikz/shorten right}}

\pgfpathmoveto{\pgfpoint{0.5 * (\pgf@xa + \pgf@xb)}{\pgf@yb +0pt}}
\pgfpathlineto{\pgfpoint{\pgf@xa - 8pt}{\pgf@ya + 3 * \pgf@shorten@left - 5pt}} 
\pgfpathlineto{\pgfpoint{\pgf@xa - 8pt}{\pgf@ya + 1pt}}
\pgfpathlineto{\pgfpoint{\pgf@xb + 8pt}{\pgf@ya + 1pt}}
\pgfpathlineto{\pgfpoint{\pgf@xb + 8pt}{\pgf@ya + 3 * \pgf@shorten@left - 5pt}}
\pgfpathclose
}
}

\pgfdeclareshape{langrect}{
\inheritsavedanchors[from=rectangle] 
\inheritanchorborder[from=rectangle]
\inheritanchor[from=rectangle]{center}
\inheritanchor[from=rectangle]{north}
\inheritanchor[from=rectangle]{south}
\inheritanchor[from=rectangle]{west}
\inheritanchor[from=rectangle]{east}
\backgroundpath{
\southwest \pgf@xa=\pgf@x \pgf@ya=\pgf@y
\northeast \pgf@xb=\pgf@x \pgf@yb=\pgf@y

\pgfmathsetmacro{\pgf@shorten@left}{\pgfkeysvalueof{/tikz/shorten left}}
\pgfmathsetmacro{\pgf@shorten@right}{\pgfkeysvalueof{/tikz/shorten right}}

\pgfpathmoveto{\pgfpoint{\pgf@xa - 8pt}{\pgf@ya + 3 * \pgf@shorten@left - 5pt}} 
\pgfpathlineto{\pgfpoint{\pgf@xa - 8pt}{\pgf@ya + 1pt}}
\pgfpathlineto{\pgfpoint{\pgf@xb + 8pt}{\pgf@ya + 1pt}}
\pgfpathlineto{\pgfpoint{\pgf@xb + 8pt}{\pgf@ya + 3 * \pgf@shorten@left - 5pt}}
\pgfpathclose
}
}

\pgfdeclareshape{langrecthigh}{
\inheritsavedanchors[from=rectangle] 
\inheritanchorborder[from=rectangle]
\inheritanchor[from=rectangle]{center}
\inheritanchor[from=rectangle]{north}
\inheritanchor[from=rectangle]{south}
\inheritanchor[from=rectangle]{west}
\inheritanchor[from=rectangle]{east}
\backgroundpath{
\southwest \pgf@xa=\pgf@x \pgf@ya=\pgf@y
\northeast \pgf@xb=\pgf@x \pgf@yb=\pgf@y

\pgfmathsetmacro{\pgf@shorten@left}{\pgfkeysvalueof{/tikz/shorten left}}
\pgfmathsetmacro{\pgf@shorten@right}{\pgfkeysvalueof{/tikz/shorten right}}

\pgfpathmoveto{\pgfpoint{\pgf@xa - 8pt}{\pgf@ya + 3 * \pgf@shorten@left - 0pt}} 
\pgfpathlineto{\pgfpoint{\pgf@xa - 8pt}{\pgf@ya + 1pt}}
\pgfpathlineto{\pgfpoint{\pgf@xb + 8pt}{\pgf@ya + 1pt}}
\pgfpathlineto{\pgfpoint{\pgf@xb + 8pt}{\pgf@ya + 3 * \pgf@shorten@left - 0pt}}
\pgfpathclose
}
}

\makeatother

\pgfkeyssetvalue{/tikz/shorten left}{0pt}
\pgfkeyssetvalue{/tikz/shorten right}{0pt}

\tikzstyle{kpoint common}=[draw,fill=white,inner sep=1pt,minimum height=4mm]

\tikzstyle{langstate}=[shape=langcopoint,shorten left=5pt,kpoint common,font=\footnotesize]
\tikzstyle{langstatehigh}=[shape=langcopointhigh,shorten left=5pt,kpoint common,font=\footnotesize]
\tikzstyle{langeffect}=[shape=langpoint,shorten left=5pt,kpoint common,font=\footnotesize]
\tikzstyle{langbox}=[shape=langrect,shorten left=5pt,kpoint common,font=\footnotesize] 
\tikzstyle{langboxhigh}=[shape=langrecthigh,shorten left=5pt,kpoint common,font=\footnotesize] 

\tikzstyle{kpoint}=[shape=cornerpoint,shorten left=5pt,kpoint common]
\tikzstyle{kpoint adjoint}=[shape=cornercopoint,shorten left=5pt,kpoint common]

\tikzstyle{kpoint conjugate}=[shape=cornerpoint,shorten right=5pt,kpoint common]
\tikzstyle{kpoint transpose}=[shape=cornercopoint,shorten right=5pt,kpoint common]
\tikzstyle{kpoint symm}=[shape=cornerpoint,shorten left=5pt,shorten right=5pt,kpoint common]

\tikzstyle{black kpoint}=[shape=cornerpoint,shorten left=5pt,kpoint common,fill=black,font=\color{white}]
\tikzstyle{black kpoint adjoint}=[shape=cornercopoint,shorten left=5pt,kpoint common,fill=black,font=\color{white}]
\tikzstyle{black kpointadj}=[shape=cornercopoint,shorten left=5pt,kpoint common,fill=black,font=\color{white}]

\tikzstyle{black dkpoint}=[shape=cornerpoint,shorten left=5pt,kpoint common,fill=black, doubled,font=\color{white}]
\tikzstyle{black dkpoint adjoint}=[shape=cornercopoint,shorten left=5pt,kpoint common,fill=black, doubled,font=\color{white}]
\tikzstyle{black dkpointadj}=[shape=cornercopoint,shorten left=5pt,kpoint common,fill=black, doubled,font=\color{white}]

\tikzstyle{kpointdag}=[kpoint adjoint]
\tikzstyle{kpointadj}=[kpoint adjoint]
\tikzstyle{kpointconj}=[kpoint conjugate]
\tikzstyle{kpointtrans}=[kpoint transpose]

\tikzstyle{big kpoint}=[kpoint, minimum width=1.2 cm, minimum height=8mm, inner sep=4pt, text depth=3mm]

\tikzstyle{wide kpoint}=[kpoint, minimum width=1 cm, inner sep=2pt]
\tikzstyle{wide kpointdag}=[kpointdag, minimum width=1 cm, inner sep=2pt]
\tikzstyle{wide kpointconj}=[kpointconj, minimum width=1 cm, inner sep=2pt]
\tikzstyle{wide kpointtrans}=[kpointtrans, minimum width=1 cm, inner sep=2pt]

\tikzstyle{gray kpoint}=[kpoint,fill=gray!50!white]
\tikzstyle{gray kpointdag}=[kpointdag,fill=gray!50!white]
\tikzstyle{gray kpointadj}=[kpointadj,fill=gray!50!white]
\tikzstyle{gray kpointconj}=[kpointconj,fill=gray!50!white]
\tikzstyle{gray kpointtrans}=[kpointtrans,fill=gray!50!white]

\tikzstyle{gray dkpoint}=[kpoint,fill=gray!50!white,doubled]
\tikzstyle{gray dkpointdag}=[kpointdag,fill=gray!50!white,doubled]
\tikzstyle{gray dkpointadj}=[kpointadj,fill=gray!50!white,doubled]
\tikzstyle{gray dkpointconj}=[kpointconj,fill=gray!50!white,doubled]
\tikzstyle{gray dkpointtrans}=[kpointtrans,fill=gray!50!white,doubled]

\tikzstyle{white label}=[draw,fill=white,rectangle,inner sep=0.7 mm]
\tikzstyle{gray label}=[draw,fill=gray!50!white,rectangle,inner sep=0.7 mm]
\tikzstyle{black label}=[draw,fill=black,rectangle,inner sep=0.7 mm]

\tikzstyle{dkpoint}=[kpoint,doubled]
\tikzstyle{wide dkpoint}=[wide kpoint,doubled]
\tikzstyle{dkpointdag}=[kpoint adjoint,doubled]
\tikzstyle{wide dkpointdag}=[wide kpointdag,doubled]
\tikzstyle{dkcopoint}=[kpoint adjoint,doubled]
\tikzstyle{dkpointadj}=[kpoint adjoint,doubled]
\tikzstyle{dkpointconj}=[kpoint conjugate,doubled]
\tikzstyle{dkpointtrans}=[kpoint transpose,doubled]

\tikzstyle{kscalar}=[kpoint common, shape=EBox, inner xsep=-1pt, inner ysep=3pt,font=\small]
\tikzstyle{kscalarconj}=[kpoint common, shape=WBox, inner xsep=-1pt, inner ysep=3pt,font=\small]

 \tikzstyle{upground}=[circuit ee IEC,ground,rotate=90,scale=2.5]
 \tikzstyle{downground}=[circuit ee IEC,ground,rotate=-90,scale=2.5]
 \tikzstyle{bigground}=[regular polygon,regular polygon sides=3,draw=gray,scale=0.50,inner sep=-0.5pt,minimum width=10mm,fill=gray]

\tikzstyle{arrs}=[-latex,font=\small,auto]
\tikzstyle{arrow plain}=[arrs]
\tikzstyle{arrow dashed}=[dashed,arrs]
\tikzstyle{arrow bold}=[very thick,arrs]
\tikzstyle{arrow hide}=[draw=white!0,-]
\tikzstyle{arrow reverse}=[latex-]
\tikzstyle{cdnode}=[]

\definecolor{hexcolor0xa9a9a9}{rgb}{0.663,0.663,0.663}  
\tikzstyle{GrayLine}=[dashed,draw=hexcolor0xa9a9a9]
\tikzstyle{gray}=[dashed,draw=hexcolor0xa9a9a9]

\theoremstyle{definition}
\newtheorem{theorem}{Theorem}[section]
\newtheorem*{theorem*}{Theorem}

\newtheorem{lemma}[theorem]{Lemma} 
\newtheorem{prop}[theorem]{Proposition}

\newtheorem{example*}[theorem]{Example*}
\newtheorem{examples*}[theorem]{Examples*}
\newtheorem{remark}[theorem]{Remark}
\newtheorem{remark*}[theorem]{Remark*}

\def\bR{\begin{color}{red}}
\def\bB{\begin{color}{blue}}
\def\bM{\begin{color}{magenta}}
\def\bC{\begin{color}{cyan}}
\def\bW{\begin{color}{white}}
\def\bMl{\begin{color}{black}}
\def\bG{\begin{color}{green}}
\def\bY{\begin{color}{yellow}}
\def\e{\end{color}\xspace}
\newcommand{\bit}{\begin{itemize}}
\newcommand{\eit}{\end{itemize}\par\noindent}
\newcommand{\ben}{\begin{enumerate}}
\newcommand{\een}{\end{enumerate}\par\noindent}
\newcommand{\beq}{\begin{equation}}
\newcommand{\eeq}{\end{equation}\par\noindent}
\newcommand{\beqa}{\begin{eqnarray*}}
\newcommand{\eeqa}{\end{eqnarray*}\par\noindent}
\newcommand{\beqn}{\begin{eqnarray}}
\newcommand{\eeqn}{\end{eqnarray}\par\noindent}

\title{Meaning updating of density matrices} 
\titlerunning{Meaning updating in DisCoCirc}  
\addauthor[coecke@cs.ox.ac.uk]{Bob Coecke} 
{Oxford University, Department of Computer Science\\ Cambridge Quantum Computing Ltd.}
\addauthor[k.mei@cambridgequantum.com]{Konstantinos Meichanetzidis} 
{Oxford University, Department of Computer Science\\ Cambridge Quantum Computing Ltd.}
\titlethanks{We thank Jon Barrett, Stefano Gogioso and Rob Spekkens for inspiring discussions preceding the results in this paper, and Matt Pusey, Sina Salek and Sam Staton for comments on a previous version of the paper. KM was supported by the EPSRC National Hub in Networked Quantum Information Technologies. 
} 

\begin{document}
\maketitle 
 

\begin{abstract}
The DisCoCat model of natural language meaning assigns meaning to a sentence given: (i) the meanings of its words, and, (ii) its grammatical structure. The recently introduced DisCoCirc model extends this to text consisting of multiple sentences. While in DisCoCat all meanings are fixed, in DisCoCirc each sentence updates meanings of words. In this paper we explore different update mechanisms for DisCoCirc, in the case where meaning is encoded in density matrices--- which come with several advantages as compared to vectors.

Our starting point are two non-commutative update mechanisms, borrowing one from quantum foundations research \cite{Leifer1, leifer2013towards}, and the other one from \cite{CoeckeText, MarthaDot}. Unfortunately, neither of these satisfies any desirable algebraic properties, nor are internal to the meaning category. By passing to double density matrices \cite{Ashoush, Zwart2017} we do get an elegant internal diagrammatic update mechanism. 
 
We also show that (commutative) spiders can be cast as an instance of the update mechanism of \cite{Leifer1, leifer2013towards}. This result is of interest to quantum foundations, as it bridges the work in Categorical Quantum Mechanics (CQM) with that on conditional quantum states. Our work also underpins implementation of text-level Natural Language Processing (NLP) on quantum hardware, for which  exponential space-gain and quadratic speed-up have previously been identified.
\end{abstract}  

\section{Intro}     

While grammar is a mathematically well studied structure \cite{Ajdukiewicz, Lambek0, Grishin, LambekBook}, in Natural Language Processing (NLP) this mathematical structure is still largely ignored. The Categorical Distributional Semantics (DisCoCat) framework \cite{CSC} was introduced in order to address this problem: it exploits grammatical structure in order to derive meanings of sentences from the meanings of its constituent words. For doing so we mostly relied on Lambek's pregroups \cite{LambekBook}, because of their simplicity, but any other mathematical model of grammar would work as well \cite{LambekvsLambek}.

In NLP, meanings are established empirically (e.g.~\cite{harris1954distributional}), and this leads to a vector space representation. DisCoCat allows for meanings to be described in a variety of models, including the vector spaces widely used in NLP \cite{CSC, GrefSadr, KartSadr}, but also relations as widely used in logic \cite{CSC}, density matrices \cite{calco2015, EsmaSC, bankova2016graded}, conceptual spaces \cite{ConcSpacI}, as well as many other more exotic models \cite{marsden2017custom, DBLP:conf/wollic/CoeckeGLM17}. 

Density matrices, which will be of interest to us in this paper, are to be conceived as an extension of the vector space model. Firstly, vector spaces do not allow for encoding lexical entailment structure such as in: 
\begin{center}
{\tt tiger $\leq$ big cat $\leq$ mammal $\leq$ vertebrate $\leq$ animal} 
\end{center}
while density matrices \cite{vN} do allow for this \cite{EsmaSC, bankova2016graded}. Density matrices have also been used in DisCoCat to encode ambiguity (a.k.a.~`lack of information') \cite{RobinMSc, DimitriDPhil, calco2015}. Here the use of density matrices perfectly matches von Neumann's motivation to introduce them for quantum theory in the first place, and why they currently also underpin quantum information theory \cite{BennettShor}. 
Density matrices also inherit the empirical benefits of vectors for NLP purposes. Other earlier uses of density matrices in NLP exploit the extended parameter space \cite{blacoe2013quantum}, which is a benefit we can also exploit.


DisCoCat does have some restrictions, however. It does not provide an obvious or unique mechanism for compositing sentences. Meanings in DisCoCat are also static, while on the other hand, in running text, meanings of words are subject to \em update\em, namely, the knowledge-update processes that the reader undergoes as they acquire more knowledge upon reading a text:
\begin{center}
{\tt Once there was Bob.\\ Bob was a dog.\\ He was a bad dog that bites.}
\end{center}
or, when properties of actors change as a story unfolds:
\begin{center}
{\tt Alice and Bob were born.\\ They got married.\\ Then they broke up.}
\end{center}
These restrictions of DisCoCat were addressed in the recently introduced DisCoCirc framework \cite{CoeckeText}, in which sentences within larger text can be composed, and meanings are updated as text progresses. This raises the new question on what these update mechanisms are for specific models. Due to the above stated motivations, we focus on meaning embeddings in density matrices.

There has been some use of meaning updating within DisCoCat, most notably, for encoding intersective adjectives \cite{ConcSpacI} and relative pronouns \cite{FrobMeanI, FrobMeanII}. Here, a property is attributed to some noun by means of a suitable connective. Thus far, DisCoCat relied on the commutative special Frobenius algebras of CQM \cite{CPV, CPaqPav}, a.k.a.~\em spiders \em \cite{CQMII, CKbook}. 
However, for the purpose of general meaning updating spiders are far  too restrictive, for example, they force updating to be commutative. For this reason in this paper we study several other update mechanisms, and provide a unified picture of these, which also encompasses spiders.

In Section \ref{sec:updatinggeneral}, we place meaning updating at the very centre of DisCoCirc: we show that DisCoCirc can be conceived as a theory about meaning updating only. This will in particular involve a representation of transitive verbs that emphasises how a verb creates a bond between the subject and the object. Such a representation has previously been used in \cite{GrefSadr, KartsaklisSadrzadeh2014}, where also experimental support was provided.

In Section \ref{sec:pedals} we identify two existing non-commutative update mechanisms for density matrices. The first one was introduced in \cite{CoeckeText, MarthaDot}, which we will refer to as \em fuzz\em, and has a very clear conceptual grounding. The other one was introduced within the context of a quantum theory of Bayesian inference \cite{Leifer1, leifer2013towards}, which we will refer to as \em phaser\em. While this update mechanism has been used in quantum foundations, and has been proposed as a connective within DisCoCat \cite{RobinMSc}, its conceptual status is much less clear, not in the least since it involves the somewhat ad hoc looking expression $\sqrt{\sigma}\, \rho\, \sqrt{\sigma}$ involving density matrices $\rho$ and $\sigma$. 
In Section \ref{sec:phaser} we show that, in fact, the phaser can be traced back to spiders, but in a manner that makes this update mechanism non-commutative. In Section \ref{sec:exp} we point at already existing experimental evidence in favour of our update mechanisms.

In Section \ref{sec:lotsofnon} we list a number of shortcomings of fuzz and phaser. Firstly, as we
have two very distinct mechanisms, the manner in which meanings get updated is not unique. Both are moreover algebraically poor (e.g.~they are non-associative). Finally, neither is internal to the meaning category of density matrices and CP-maps.

As both update mechanisms do have a natural place within a theory of meaning updating, in Section \ref{sec:unification} we propose a mechanism that has fuzz and phaser as special cases. 
We achieve this by meanings and verbs as double density matrices \cite{Ashoush, Zwart2017}, which have a richer structure than density matrices. Doing so we still remain internal to the meaning category of density matrices and CP-maps, as we demonstrate in Section \ref{sec:implementation}, where we also discuss implementation on quantum hardware.

In Section \ref{sec:examples} we provide some very simple illustrative  examples.
 
\section{Preliminaries} 

We expect the reader to have some familiarity with the DisCoCat framework \cite{CSC, FrobMeanI, CLM}, and with its diagrammatic formalism that we borrowed from Categorical Quantum Mechanics (CQM) \cite{CKpaperI, CQMII, CKbook}, most notably caps/cups, spiders, and doubling. We also expect the reader to be familiar with Dirac notation, projectors, density matrices, spectral decomposition and completely positive maps as used in quantum theory. We now set out the specific notational conventions that we will be following in this paper.

We read diagrams from top to bottom. Sometimes the boxes will represent linear maps, and sometimes they will represent completely positive maps. In order to distinguish these two representations we follow the conventions of \cite{CKbook, CQMII}, which means that a vector and a density matrix will respectively be represented as:
\ctikzfig{stateSingleDouble}
where wires represent systems. A privileged vector for two systems is the \em cap\em:
\[
\begin{tikzpicture}[scale=0.80]
	\begin{pgfonlayer}{nodelayer}
		\node [style=none] (0) at (1.25, -0.5) {};
		\node [style=none] (1) at (-1.25, -0.5) {};
	\end{pgfonlayer}
	\begin{pgfonlayer}{edgelayer}
		\draw [style=swap, bend left=90, looseness=1.50] (1.center) to (0.center);
	\end{pgfonlayer}
\end{tikzpicture}\ \ :=\ \ \sum_i |i i\rangle 
\]
and its adjoint (a.k.a.~`bra') is the \em cup\em: 
\[
\begin{tikzpicture}[scale=0.80]
	\begin{pgfonlayer}{nodelayer}
		\node [style=none] (0) at (1.25, 0.5) {};
		\node [style=none] (1) at (-1.25, 0.5) {};
	\end{pgfonlayer}
	\begin{pgfonlayer}{edgelayer}
		\draw [style=swap, bend right=90, looseness=1.50] (1.center) to (0.center);
	\end{pgfonlayer}
\end{tikzpicture}\ \ := \ \ \sum_i \langle i i|
\]

Similarly to states, linear maps and CP maps are respectively depicted as: 
\ctikzfig{boxSingleDouble}
We will reserve white dots to represent \em spiders\em:
\beq\label{eq:white spider}
\begin{tikzpicture}[scale=0.80]
	\begin{pgfonlayer}{nodelayer}
		\node [style=none] (0) at (1, -1.25) {};
		\node [style=none] (1) at (-1, -1.25) {};
		\node [style=none] (2) at (-1, 1.25) {};
		\node [style=none] (3) at (1, 1.25) {};
		\node [style=none] (4) at (0, 1) {\ldots};
		\node [style=none] (5) at (0, -1) {\ldots};
		\node [style=white dot] (6) at (0, 0) {};
		\node [style=none] (7) at (0, 0) {};
	\end{pgfonlayer}
	\begin{pgfonlayer}{edgelayer}
		\draw [style=swap, in=135, out=-90, looseness=1.00] (2.center) to (6);
		\draw [style=swap, in=45, out=-90, looseness=1.00] (3.center) to (6);
		\draw [style=swap, in=90, out=-45, looseness=1.00] (6) to (0.center);
		\draw [style=swap, in=90, out=-135, looseness=1.00] (6) to (1.center); 
	\end{pgfonlayer}
\end{tikzpicture}\ \ := \ \ \sum_i |i\ldots i\rangle\langle i\ldots i| 
\eeq
Crucially, spiders are clearly basis-dependent, and in fact, they represent orthonormal bases \cite{CPV}. Note also that caps and cups are instances of spiders, and more generally, that spiders can be conceived as `multi-wires' \cite{CKbook, CQMII}: the only thing that matters is what is connected to what by means of possibly multiple spiders, and not what the precise shape is of the spider-web that does so. This behaviour can be succinctly captured within the following \em fusion \em equation:
\beq\label{eq:fusion}
\begin{tikzpicture}[scale=0.80]
	\begin{pgfonlayer}{nodelayer}
		\node [style=none] (0) at (0, 0) {};
		\node [style=none] (1) at (-2, 0) {};
		\node [style=none] (2) at (-2, 2.5) {};
		\node [style=none] (3) at (0, 2.5) {};
		\node [style=none] (4) at (-1, 2.25) {\ldots};
		\node [style=none] (5) at (-1, 0.25) {\ldots};
		\node [style=white dot] (6) at (-1, 1.25) {};
		\node [style=none] (7) at (-1, 1.25) {};
		\node [style=none] (8) at (0, -2.5) {};
		\node [style=none] (9) at (1, -1.25) {};
		\node [style=none] (10) at (0, 0) {};
		\node [style=none] (11) at (1, -0.25) {\ldots};
		\node [style=none] (12) at (1, -2.25) {\ldots};
		\node [style=white dot] (13) at (1, -1.25) {};
		\node [style=none] (14) at (2, -2.5) {};
		\node [style=none] (15) at (2, 0) {};
		\node [style=none] (16) at (-2, -2.5) {};
		\node [style=none] (17) at (2, 2.5) {};
	\end{pgfonlayer}
	\begin{pgfonlayer}{edgelayer}
		\draw [style=swap, in=135, out=-90, looseness=1.00] (2.center) to (6);
		\draw [style=swap, in=45, out=-90, looseness=1.00] (3.center) to (6);
		\draw [style=swap, in=90, out=-45, looseness=1.00] (6) to (0.center);
		\draw [style=swap, in=90, out=-135, looseness=1.00] (6) to (1.center);
		\draw [style=swap, in=135, out=-90, looseness=1.00] (10.center) to (13);
		\draw [style=swap, in=45, out=-90, looseness=1.00] (15.center) to (13);
		\draw [style=swap, in=90, out=-45, looseness=1.00] (13) to (14.center);
		\draw [style=swap, in=90, out=-135, looseness=1.00] (13) to (8.center);
		\draw [style=swap] (17.center) to (15.center);
		\draw [style=swap] (1.center) to (16.center);
	\end{pgfonlayer}
\end{tikzpicture}\ \ =\ \ \begin{tikzpicture}[scale=0.80]
	\begin{pgfonlayer}{nodelayer}
		\node [style=none] (0) at (1, -1.25) {};
		\node [style=none] (1) at (-1, -1.25) {};
		\node [style=none] (2) at (-1, 1.25) {};
		\node [style=none] (3) at (1, 1.25) {};
		\node [style=none] (4) at (0, 1) {\ldots};
		\node [style=none] (5) at (0, -1) {\ldots};
		\node [style=white dot] (6) at (0, 0) {};
		\node [style=none] (7) at (0, 0) {};
	\end{pgfonlayer}
	\begin{pgfonlayer}{edgelayer}
		\draw [style=swap, in=135, out=-90, looseness=1.00] (2.center) to (6);
		\draw [style=swap, in=45, out=-90, looseness=1.00] (3.center) to (6);
		\draw [style=swap, in=90, out=-45, looseness=1.00] (6) to (0.center);
		\draw [style=swap, in=90, out=-135, looseness=1.00] (6) to (1.center); 
	\end{pgfonlayer}
\end{tikzpicture}
\eeq
 
By \em un-doubling \em we refer to re-writing a CP-map as follows \cite{SelingerCPM, CKbook}:
\[
\begin{tikzpicture}[scale=0.80]
	\begin{pgfonlayer}{nodelayer}
		\node [style=langbox, boldedge] (0) at (0, 0) {\!\!{\tt CP}\!\!};
		\node [style=none] (1) at (0, -0.25) {};
		\node [style=none] (2) at (0, -1.25) {};
		\node [style=none] (3) at (0, 0.25) {};
		\node [style=none] (4) at (0, 1) {};
	\end{pgfonlayer}
	\begin{pgfonlayer}{edgelayer}
		\draw [style=boldedge] (2.center) to (1.center);
		\draw [style=boldedge] (3.center) to (4.center);
	\end{pgfonlayer}
\end{tikzpicture}\ \ \leadsto\ \ \begin{tikzpicture}[scale=0.80]
	\begin{pgfonlayer}{nodelayer}
		\node [style=langbox] (0) at (-1.25, 0) {\!$f$\!};
		\node [style=none] (1) at (-1.75, -0.25) {};
		\node [style=none] (2) at (-1.75, -1.25) {};
		\node [style=none] (3) at (-1.75, 0.25) {};
		\node [style=none] (4) at (-1.75, 1) {};
		\node [style=none] (5) at (1.75, 0.25) {};
		\node [style=none] (6) at (1.75, 1) {};
		\node [style=none] (7) at (1.75, -1.25) {};
		\node [style=langbox] (8) at (1.25, 0) {\!$f$\!};
		\node [style=none] (9) at (1.75, -0.25) {};
		\node [style=none] (10) at (-0.75, -0.25) {};
		\node [style=none] (11) at (0.75, -0.25) {};
	\end{pgfonlayer}
	\begin{pgfonlayer}{edgelayer}
		\draw [style=swap] (2.center) to (1.center);
		\draw [style=swap] (3.center) to (4.center);
		\draw [style=swap] (7.center) to (9.center);
		\draw [style=swap] (5.center) to (6.center);
		\draw [style=swap, bend right=90, looseness=1.25] (10.center) to (11.center);
	\end{pgfonlayer}
\end{tikzpicture}
\]
with the two boxes being $\sum_i f_i\otimes | i \rangle$ and $\sum_i | i \rangle\otimes\bar{f}_i$ for Kraus maps $f_i$, that is:
\beq\label{eq:CPM4}
\begin{tikzpicture}[scale=0.80]
	\begin{pgfonlayer}{nodelayer}
		\node [style=none] (0) at (-3.5, -0.25) {};
		\node [style=none] (1) at (-3.5, -1.5) {};
		\node [style=none] (2) at (3.5, -1.5) {};
		\node [style=none] (3) at (3.5, -0.25) {};
		\node [style=none] (4) at (-1.5, -0.25) {};
		\node [style=none] (5) at (1.5, -0.25) {};
		\node [style=langboxhigh] (6) at (-2.5, 0) {\!\!\!\!\!\raisebox{1.5mm}{$\sum_i f_i\otimes | i \rangle$}\!\!\!\!\!};
		\node [style=langboxhigh] (7) at (2.5, 0) {\!\!\!\!\!\raisebox{1.5mm}{$\sum_i | i \rangle\otimes\bar{f}_i$}\!\!\!\!\!};
		\node [style=none] (8) at (3.5, 0.5) {};
		\node [style=none] (9) at (3.5, 1.5) {};
		\node [style=none] (10) at (-3.5, 1.5) {};
		\node [style=none] (11) at (-3.5, 0.5) {};
	\end{pgfonlayer}
	\begin{pgfonlayer}{edgelayer}
		\draw [style=swap] (1.center) to (0.center);
		\draw [style=swap] (2.center) to (3.center);
		\draw [style=swap, bend right=90, looseness=1.00] (4.center) to (5.center);
		\draw [style=swap] (11.center) to (10.center);
		\draw [style=swap] (8.center) to (9.center);
	\end{pgfonlayer}
\end{tikzpicture}
\eeq
In the specific case of density matrices this becomes: 
\beq\label{eq:CPonly2}
\begin{tikzpicture}[scale=0.80]
	\begin{pgfonlayer}{nodelayer}
		\node [style=none] (0) at (0, 0) {};
		\node [style=langstate, boldedge] (1) at (0, 0.25) {\!\!\!\!\!{\tt dens}\!\!\!\!\!};
		\node [style=none] (2) at (0, -1) {};
	\end{pgfonlayer}
	\begin{pgfonlayer}{edgelayer}
		\draw [style=boldedge] (2.center) to (0.center); 
	\end{pgfonlayer}
\end{tikzpicture} \ \ \leadsto\ \ \begin{tikzpicture}[scale=0.80]
	\begin{pgfonlayer}{nodelayer}
		\node [style=none] (0) at (-1.75, -0.25) {};
		\node [style=none] (1) at (-1.75, -1.25) {};
		\node [style=none] (2) at (1.75, -1.25) {};
		\node [style=none] (3) at (1.75, -0.25) {};
		\node [style=none] (4) at (-0.75, -0.25) {};
		\node [style=none] (5) at (0.75, -0.25) {};
		\node [style=langstate] (6) at (-1.25, 0) {\!$\omega$\!}; 
		\node [style=langstate] (7) at (1.25, 0) {\!$\omega$\!};
	\end{pgfonlayer}
	\begin{pgfonlayer}{edgelayer}
		\draw [style=swap] (1.center) to (0.center);
		\draw [style=swap] (2.center) to (3.center);
		\draw [style=swap, bend right=90, looseness=1.25] (4.center) to (5.center);
	\end{pgfonlayer}
\end{tikzpicture}
\eeq
Concretely, for a density matrix $\sum_i p_i |i\rangle\langle i|$ we have $\omega = \sum_i\sqrt{p_i} |i i \rangle$, i.e.:
\[
\begin{tikzpicture}[scale=0.80]
	\begin{pgfonlayer}{nodelayer}
		\node [style=none] (0) at (-3.5, -0.25) {};
		\node [style=none] (1) at (-3.5, -1.5) {};
		\node [style=none] (2) at (3.5, -1.5) {};
		\node [style=none] (3) at (3.5, -0.25) {};
		\node [style=none] (4) at (-1.5, -0.25) {};
		\node [style=none] (5) at (1.5, -0.25) {};
		\node [style=langstatehigh] (6) at (-2.5, 0) {\!\!\!\!\!\raisebox{1.5mm}{$\sum_i \sqrt{p_i}\, |i i \rangle$}\!\!\!\!\!};
		\node [style=langstatehigh] (7) at (2.5, 0) {\!\!\!\!\!\raisebox{1.5mm}{$\sum_i \sqrt{p_i}\, |i i \rangle$}\!\!\!\!\!};
	\end{pgfonlayer}
	\begin{pgfonlayer}{edgelayer}
		\draw [style=swap] (1.center) to (0.center);
		\draw [style=swap] (2.center) to (3.center);
		\draw [style=swap, bend right=90, looseness=1.00] (4.center) to (5.center);
	\end{pgfonlayer}
\end{tikzpicture}
\] 
 
\section{Text meaning in DisCoCirc as updating}\label{sec:updatinggeneral}

The starting point of both DisCoCat and DisCoCirc is the fact that pregroup analysis \cite{LambekBook} of the grammatical structure associates to each sentences a diagram. In the case of a sentence (of which the associated grammatical type is denoted $s$) consisting of a subject (with type $n$ for `noun'), a transitive verb (with composite type ${}^{-1}n \cdot s \cdot n^{-1}$), and an object (also with type $n$) this diagram looks as follows:
\ctikzfig{hates}
In DisCoCat \cite{CSC} we then replace the types by the encoding of the word meanings,
\ctikzfig{s4copy}
which in our case are represented by density matrices. The wires are interpreted as maps, for example, the cups will be the CP-maps associated to Bell-effects:
\[
\sum_i \langle ii |
\]
while the straight wire is an identity.

The above assumes that all words have fixed meanings given by those density matrices.
However, as already explained in the introduction, meanings evolve in the course of developing text. Therefore, in DisCoCirc \cite{CoeckeText}, prime actors like {\tt Bob} and {\tt Alice} in {\tt Bob bites Alice} are not represented by a state, but instead by a wire carrying the evolving meaning.
For this purpose we take the $s$-type to be the same as the type of that of the incoming actors \cite{CoeckeText}: 
\ctikzfig{s5copy}
If we happen to have prior knowledge about that actor we can always plug a corresponding `prior' state at the input of the wires: 
\ctikzfig{s5copytris}
which yields a DisCoCat-style picture. One can think of a sentence with an open input as a function $f$, while providing a prior state corresponds to $f$ being applied to a concrete input as in $f(x)$.
 The major advantage of not fixing states as a default is that this now allows us to compose sentences, for example: 
\ctikzfig{s5copybis}

The above in particular means that the $s$-type will depend on the sentence. For example, in {\tt Bob is (a) dog}, the noun {\tt dog} can be taken to be fixed, so that the $s$-type becomes the {\tt Bob}-wire alone, and we also introduce a special notation for {\tt is}: 
\beq\label{eq:preupdate}
\begin{tikzpicture}[scale=0.80]
	\begin{pgfonlayer}{nodelayer}
		\node [style=none] (0) at (-3.5, -0.25) {};
		\node [style=none] (1) at (-1, -0.25) {};
		\node [style=none] (2) at (0, -0.25) {};
		\node [style=none] (3) at (3.5, -0.25) {};
		\node [style=none] (4) at (1, -0.25) {};
		\node [style=none] (5) at (0, -1.5) {};
		\node [style=gray dot, boldedge] (6) at (0, 0.5) {};
		\node [style=none] (7) at (1.25, -0.25) {};
		\node [style=none] (8) at (-1.25, -0.25) {};
		\node [style=none] (9) at (1.25, 0.5) {};
		\node [style=none] (10) at (-1.25, 0.5) {}; 
		\node [style=none] (11) at (0, 1) {};
		\node [style=none] (12) at (-3.5, 1.75) {\!\!\!{\footnotesize\tt Bob}\!\!\!};
		\node [style=langstate, boldedge] (13) at (3.5, 0.25) {\!\!\!{\tt dog}\!\!\!};
		\node [style=none] (14) at (0, 1.75) {\!\!\!{\footnotesize\tt is}\!\!\!};
		\node [style=none] (15) at (-3.5, 1) {};
	\end{pgfonlayer}
	\begin{pgfonlayer}{edgelayer}
		\draw [style=boldedge, bend right=90, looseness=1.25] (0.center) to (1.center);
		\draw [style=boldedge, bend right=90, looseness=1.25] (4.center) to (3.center);
		\draw [style=boldedge] (2.center) to (5.center);
		\draw [style=boldedge, in=-176, out=90, looseness=1.00] (1.center) to (6); 
		\draw [style=swap, dashed] (10.center) to (8.center);
		\draw [style=swap, dashed] (8.center) to (7.center);
		\draw [style=swap, dashed] (7.center) to (9.center);
		\draw [style=swap, dashed] (11.center) to (10.center);
		\draw [style=swap, dashed] (11.center) to (9.center);
		\draw [style=boldedge] (6) to (2.center);
		\draw [style=boldedge, in=90, out=0, looseness=1.00] (6) to (4.center);
		\draw [style=boldedge] (15.center) to (0.center);
	\end{pgfonlayer}
\end{tikzpicture} 
\eeq
Yanking wires this becomes: 
\beq\label{eq:update}
\begin{tikzpicture}[scale=0.80]
	\begin{pgfonlayer}{nodelayer}
		\node [style=none] (0) at (-0.5, -2) {};
		\node [style=none] (1) at (-0.5, 1.25) {};
		\node [style=none] (2) at (-0.5, 1.75) {\footnotesize\tt Bob};
		\node [style=langstate, boldedge] (3) at (1, 0.25) {\!\!\!{\tt dog}\!\!\!};
		\node [style=gray dot, boldedge] (4) at (-0.5, -1.25) {};
	\end{pgfonlayer}
	\begin{pgfonlayer}{edgelayer}
		\draw [style=boldedge, in=45, out=-90, looseness=1.00] (3) to (4); 
		\draw [style=boldedge] (4) to (0.center);
		\draw [style=boldedge] (1.center) to (4); 
	\end{pgfonlayer}
\end{tikzpicture}
\eeq
Here we think of {\tt dog} as an adjective, for {\tt Bob}. This reflects what we mean by an \em update mechanism \em in this paper: we update an actor's meaning (here {\tt Bob}) by imposing a feature (here {\tt Dog}) by means of the grey dot connective, where by `actor' we refer to varying nouns subject to update. 

We can also put more general transitive verbs into an adjective-like shape like in (\ref{eq:update}) by using the verb-form introduced in \cite{GrefSadr, KartsaklisSadrzadeh2014}:
\[
\begin{tikzpicture}[scale=0.80]
	\begin{pgfonlayer}{nodelayer}
		\node [style=none] (0) at (-0.25, -1.75) {};
		\node [style=gray dot, boldedge] (1) at (-1, -1) {};
		\node [style=none] (2) at (0.25, -1.75) {};
		\node [style=none] (3) at (1, 0.25) {};
		\node [style=none] (4) at (-1, 0.25) {};
		\node [style=gray dot, boldedge] (5) at (1, -1) {};
		\node [style=none] (6) at (-1.75, -1.75) {};
		\node [style=none] (7) at (0.25, -2.75) {};
		\node [style=none] (8) at (1.75, -1.75) {};
		\node [style=none] (9) at (0.25, -1.75) {};
		\node [style=none] (10) at (-0.25, -1.75) {};
		\node [style=none] (11) at (-0.25, -2.75) {};
		\node [style=none] (12) at (-2, 0.75) {};
		\node [style=none] (13) at (2, 0.75) {};
		\node [style=none] (14) at (2, -1.75) {};
		\node [style=none] (15) at (0, 1) {};
		\node [style=none] (16) at (-2, -1.75) {};
		\node [style=langstate, boldedge] (17) at (0, 0.25) {\!\!\!{\tt bites}\!\!\!}; 
		\node [style=none] (18) at (3.5, 1.25) {};
		\node [style=none] (19) at (1.75, -1.75) {};
		\node [style=none] (20) at (3.5, -1.75) {};
		\node [style=none] (21) at (-3.5, -1.75) {};
		\node [style=none] (22) at (3.5, 1.75) {\footnotesize\tt  Alice};
		\node [style=none] (23) at (-3.5, 1.25) {};
		\node [style=none] (24) at (-1.75, -1.75) {};
		\node [style=none] (25) at (-3.5, 1.75) {\footnotesize\tt  Bob};
		\node [style=label] (26) at (0, 2.5) {\color{gray} verb};
		\node [style=none] (27) at (0, 1.25) {};
		\node [style=none] (28) at (0, 2) {};
	\end{pgfonlayer}
	\begin{pgfonlayer}{edgelayer}
		\draw [style=boldedge, bend left=45, looseness=1.25] (5) to (8.center);
		\draw [style=boldedge, bend right=45, looseness=1.25] (5) to (2.center);
		\draw [style=boldedge, bend left=45, looseness=1.25] (6.center) to (1);
		\draw [style=boldedge, bend left=45, looseness=1.25] (1) to (10.center);
		\draw [style=boldedge] (3) to (5);
		\draw [style=boldedge] (4) to (1);
		\draw [style=boldedge] (0.center) to (11.center);
		\draw [style=boldedge] (9.center) to (7.center);
		\draw [style=swap, dashed] (12.center) to (16.center);
		\draw [style=swap, dashed] (16.center) to (14.center);
		\draw [style=swap, dashed] (14.center) to (13.center);
		\draw [style=swap, dashed] (15.center) to (12.center);
		\draw [style=swap, dashed] (15.center) to (13.center);
		\draw [style=boldedge, bend right=90, looseness=1.25] (21.center) to (24.center);
		\draw [style=boldedge, bend right=90, looseness=1.25] (19.center) to (20.center);
		\draw [style=boldedge] (23.center) to (21.center);
		\draw [style=boldedge] (18.center) to (20.center);
		\draw [style=pointer edge, in=90, out=-90, looseness=1.25] (28.center) to (27.center);
	\end{pgfonlayer}
\end{tikzpicture}\ \ =\ \ \begin{tikzpicture}[scale=0.80]
	\begin{pgfonlayer}{nodelayer}
		\node [style=gray dot, boldedge] (0) at (-2.75, -2) {};
		\node [style=none] (1) at (2.75, 0) {};
		\node [style=none] (2) at (0.75, 0) {};
		\node [style=gray dot, boldedge] (3) at (-0.25, -2) {};
		\node [style=langstate, boldedge] (4) at (1.75, 0.25) {\!\!\!{\tt bites}\!\!\!};
		\node [style=none] (5) at (-2.75, -2.75) {};
		\node [style=none] (6) at (-2.75, 1.25) {};
		\node [style=none] (7) at (-2.75, 1.75) {\footnotesize\tt  Bob};
		\node [style=none] (8) at (-0.25, 1.75) {\footnotesize\tt  Alice};
		\node [style=none] (9) at (-0.25, -2.75) {};
		\node [style=none] (10) at (-0.25, 1.25) {};
		\node [style=none] (11) at (-1, -2.75) {};
		\node [style=none] (12) at (-1, 2.75) {};
	\end{pgfonlayer}
	\begin{pgfonlayer}{edgelayer}
		\draw [style=boldedge, in=135, out=-90, looseness=1.00] (1) to (3);
		\draw [style=boldedge, in=45, out=-90, looseness=0.75] (2) to (0);
		\draw [style=boldedge] (6.center) to (5.center);
		\draw [style=boldedge] (10.center) to (9.center);
	\end{pgfonlayer}
\end{tikzpicture}
\]
This representation of transitive verbs emphasises how a verb creates a bond between the subject and the object.
From this point of view,
an entire text 
can in principle be reduced to updates of this form,
with the grey dot playing a central role. 

The main remaining question now is: 
\begin{center}
\em What is the grey dot? \em
\end{center}
A first obvious candidate are the spiders (\ref{eq:white spider}), which have been previously employed in DisCoCat
for intersective adjectives \cite{ConcSpacI} and relative pronouns \cite{FrobMeanI, FrobMeanII}.  However, while by spiders being multi-wires, by fusion (\ref{eq:fusion}) we have:
\[
\begin{tikzpicture}[scale=0.80]
	\begin{pgfonlayer}{nodelayer}
		\node [style=none] (0) at (-1, 1.5) {};
		\node [style=white dot, boldedge] (1) at (0, 0.25) {};
		\node [style=none] (2) at (0, -2) {};
		\node [style=none] (3) at (0, 2) {};
		\node [style=none] (4) at (1, 0) {};
		\node [style=white dot, boldedge] (5) at (0, -1.25) {};
	\end{pgfonlayer}
	\begin{pgfonlayer}{edgelayer}
		\draw [style=boldedge, in=135, out=-90, looseness=0.75] (0) to (1);
		\draw [style=boldedge] (3.center) to (2.center);
		\draw [style=boldedge, in=45, out=-90, looseness=0.75] (4) to (5); 
	\end{pgfonlayer}
\end{tikzpicture}\ \ =\ \ \begin{tikzpicture}[scale=0.80]
	\begin{pgfonlayer}{nodelayer}
		\node [style=none] (0) at (-1, 0) {};
		\node [style=white dot, boldedge] (1) at (0, -1.25) {};
		\node [style=none] (2) at (0, -2) {};
		\node [style=none] (3) at (0, 2) {};
		\node [style=none] (4) at (1, 1.5) {};
		\node [style=white dot, boldedge] (5) at (0, 0.25) {};
	\end{pgfonlayer}
	\begin{pgfonlayer}{edgelayer}
		\draw [style=boldedge, in=135, out=-90, looseness=0.75] (0) to (1);
		\draw [style=boldedge] (3.center) to (2.center);
		\draw [style=boldedge, in=45, out=-90, looseness=0.75] (4) to (5); 
	\end{pgfonlayer}
\end{tikzpicture}
\]
clearly we don't have: 
\[
\begin{tikzpicture}[scale=0.80]
	\begin{pgfonlayer}{nodelayer}
		\node [style=white dot, boldedge] (0) at (-3.5, -0.5) {};
		\node [style=none] (1) at (-0.5, 0.75) {};
		\node [style=none] (2) at (-2.5, 0.75) {};
		\node [style=white dot, boldedge] (3) at (0.5, -0.5) {};
		\node [style=langstate, boldedge] (4) at (-1.5, 1) {\!\!\!{\tt bites}\!\!\!};
		\node [style=none] (5) at (-3.5, -3.5) {};
		\node [style=none] (6) at (-3.5, 1.75) {};
		\node [style=none] (7) at (-3.5, 2.25) {\footnotesize\tt  Alice};
		\node [style=none] (8) at (0.5, 2.25) {\footnotesize\tt  Bob};
		\node [style=none] (9) at (0.5, -3.5) {};
		\node [style=none] (10) at (0.5, 1.75) {};
		\node [style=langstate, boldedge] (11) at (2, -1.25) {\!\!\!{\tt sad}\!\!\!};
		\node [style=white dot, boldedge] (12) at (0.5, -2.75) {};
	\end{pgfonlayer}
	\begin{pgfonlayer}{edgelayer}
		\draw [style=boldedge, in=135, out=-90, looseness=0.75] (1) to (3);
		\draw [style=boldedge, in=45, out=-90, looseness=0.75] (2) to (0);
		\draw [style=boldedge] (6.center) to (5.center);
		\draw [style=boldedge] (10.center) to (9.center);
		\draw [style=boldedge, in=45, out=-90, looseness=1.00] (11) to (12); 
	\end{pgfonlayer}
\end{tikzpicture}\ \ \not=\ \ \begin{tikzpicture}[scale=0.80]
	\begin{pgfonlayer}{nodelayer}
		\node [style=white dot, boldedge] (0) at (-3.5, -2.75) {};
		\node [style=none] (1) at (-0.5, -1.5) {};
		\node [style=none] (2) at (-2.5, -1.5) {};
		\node [style=white dot, boldedge] (3) at (0.5, -2.75) {};
		\node [style=langstate, boldedge] (4) at (-1.5, -1.25) {\!\!\!{\tt bites}\!\!\!};
		\node [style=none] (5) at (-3.5, -3.5) {};
		\node [style=none] (6) at (-3.5, 1.75) {};
		\node [style=none] (7) at (-3.5, 2.25) {\footnotesize\tt  Alice};
		\node [style=none] (8) at (0.5, 2.25) {\footnotesize\tt  Bob};
		\node [style=none] (9) at (0.5, -3.5) {};
		\node [style=none] (10) at (0.5, 1.75) {};
		\node [style=langstate, boldedge] (11) at (2, 1) {\!\!\!{\tt sad}\!\!\!};
		\node [style=white dot, boldedge] (12) at (0.5, -0.5) {};
	\end{pgfonlayer}
	\begin{pgfonlayer}{edgelayer}
		\draw [style=boldedge, in=135, out=-90, looseness=0.75] (1) to (3);
		\draw [style=boldedge, in=45, out=-90, looseness=0.75] (2) to (0); 
		\draw [style=boldedge] (6.center) to (5.center);
		\draw [style=boldedge] (10.center) to (9.center);
		\draw [style=boldedge, in=45, out=-90, looseness=1.00] (11) to (12);
	\end{pgfonlayer}
\end{tikzpicture}
\]
Therefore, in this case the grey dot needs to be something else. As spiders do make sense in certain cases, we desire something that has spiders as a special case, and as we shall see, this wish will be fulfilled.

\section{Updating density matrices: fuzz vs.~phaser}\label{sec:pedals}

What is the most basic form of update for density matrices? Following Birkhoff-von Neumann quantum logic \cite{BvN}, any proposition about a physical system corresponds to a subspace $A$, or its corresponding projector $P_A$. Following \cite{widdows2003word, Widdows} it is also natural to think of propositions for natural language meaning like {\tt (being a) dog} as such a projector. Imposing a proposition on a density matrix is then realised as follows:
\beq\label{eq:updateP} 
P\circ - \circ P
\eeq
for example, $P_{\tt dog}\circ \rho_{\tt Bob} \circ P_{\tt dog}$. Typically the resulting density matrix won't be normalised, so we will use the term density matrix also for sub-normalised and super-normalised positive matrices.

Now, representing meanings by density matrices also {\tt dog} itself would also correspond to a density matrix in (\ref{eq:update}) for an appropriate choice of the grey dot. Fortunately, 
each projector $P$ is a (super-normalised) density matrix.

More generally, by means of weighted sums of projectors we obtain general density matrices in the form of their spectral decomposition:
\beq\label{eq:speccom}
\sum_i x_i P_i
\eeq
where one could imagine these sums to arise from the specific empirical procedure (e.g.~\cite{calco2015, MarthaNeg}) that is used to establish the meaning of {\tt dog}. With {\tt dog} itself a density matrix, we can now think of the grey dot in (\ref{eq:update}) as combining two density matrices:
\beq\label{s11bisbis1}
\begin{tikzpicture}[scale=0.80]
	\begin{pgfonlayer}{nodelayer}
		\node [style=none] (0) at (0, -1) {};
		\node [style=langstate, boldedge] (1) at (0, 0.25) {\!\!\!{\tt Bob is (a) dog}\!\!\!};
	\end{pgfonlayer}
	\begin{pgfonlayer}{edgelayer}
		\draw [style=boldedge] (1) to (0.center);
	\end{pgfonlayer}
\end{tikzpicture}\ \ := \ \ \begin{tikzpicture}[scale=0.80]
	\begin{pgfonlayer}{nodelayer}
		\node [style=none] (0) at (-0.5, -1.5) {};
		\node [style=langstate, boldedge] (1) at (1, 0.75) {\!\!\!{\tt dog}\!\!\!};
		\node [style=gray dot, boldedge] (2) at (-0.5, -0.75) {};
		\node [style=langstate, boldedge] (3) at (-2, 0.75) {\!\!\!{\tt Bob}\!\!\!};
	\end{pgfonlayer}
	\begin{pgfonlayer}{edgelayer}
		\draw [style=boldedge, in=45, out=-90, looseness=1.00] (1) to (2);
		\draw [style=boldedge] (2) to (0.center);
		\draw [style=boldedge, in=135, out=-90, looseness=1.00] (3) to (2);
	\end{pgfonlayer}
\end{tikzpicture}
\eeq
We now consider some candidates for such a grey dot. Firstly, let's eliminate two candidates. 
Composing two density matrices by matrix multiplication doesn't in general return a density matrix, nor does this have an operational interpretation. Alternatively, component-wise multiplication corresponds to fusion via spiders,
which as discussed above is too specialised as it is commutative.

Two alternatives for these have already appeared in the literature:
\begin{align}\label{pedal1} 
\rho \,\raisebox{-1.5mm}{\epsfig{figure=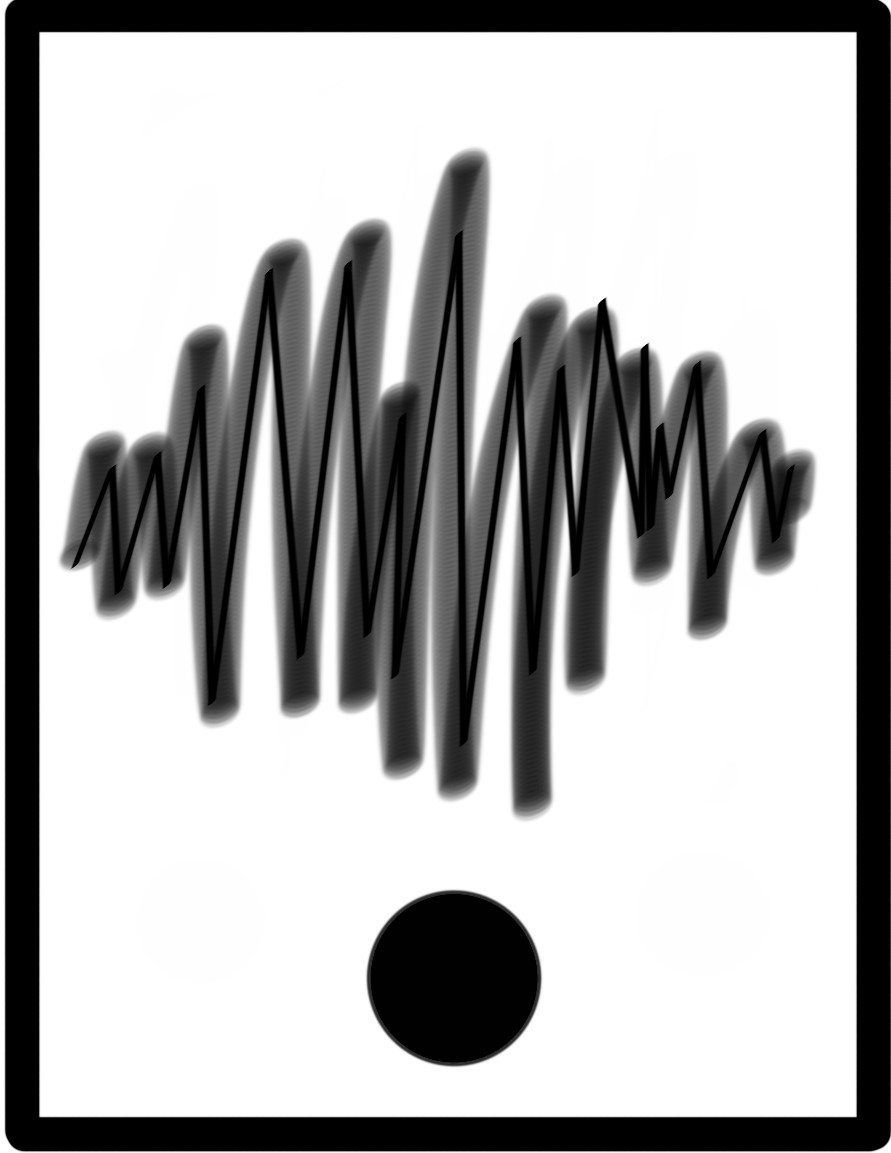,width=11pt}}\, \sigma 
&:= \sum_i x_i \Bigl(P_i\circ \rho \circ P_i\Bigr)
& 
&\mbox{with} 
& 
\sigma 
&:=\sum_i x_i P_i 
\\ \label{pedal2}
\rho \,\raisebox{-1.5mm}{\epsfig{figure=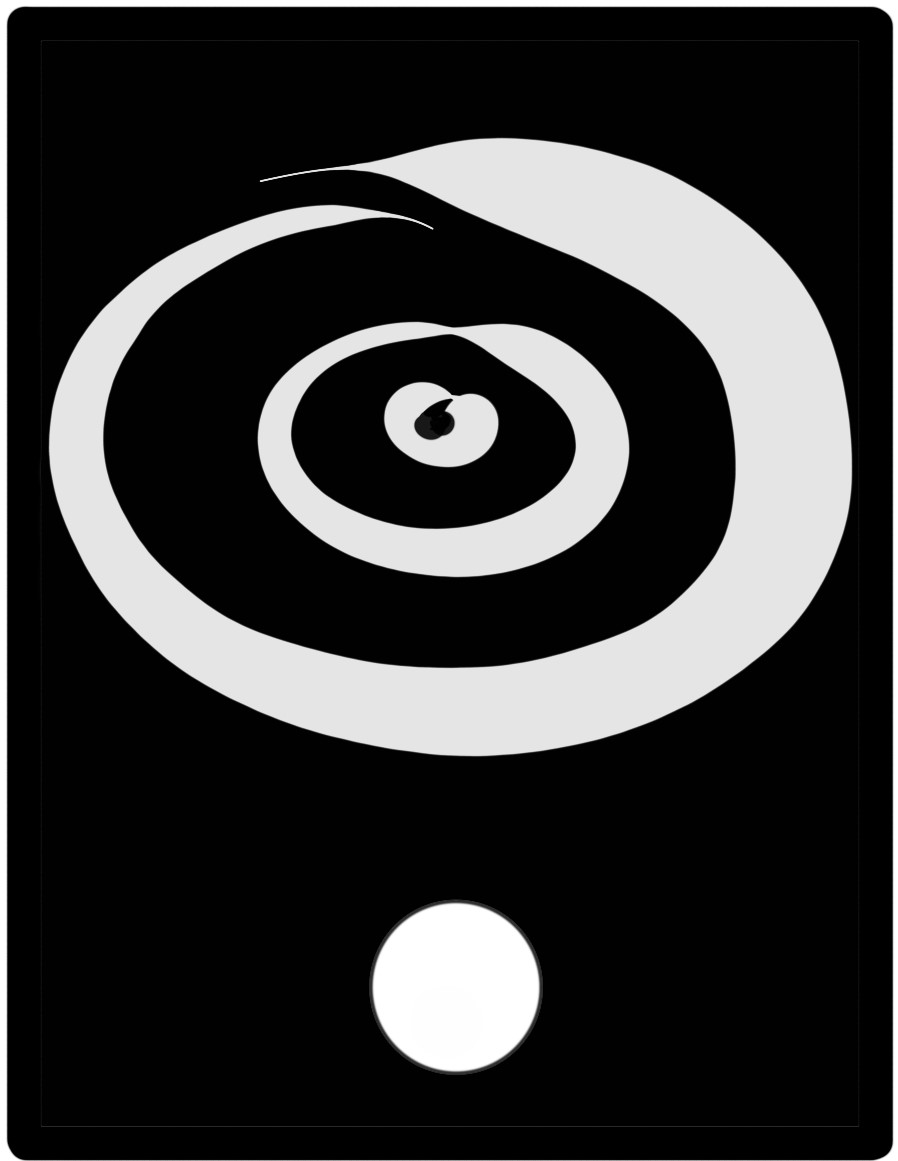,width=11pt}}\, \sigma 
&:= \Biggl(\sum_i x_i P_i\Biggr)\circ \rho \circ \Biggl(\sum_j x_j P_j\Biggr)
&
&\mbox{with}
& 
\sigma 
&:=\sum_i x_i^2 P_i
\end{align}
where for each we introduced a new dedicated `guitar-pedal' notation. 
The first of these was proposed in \cite{CoeckeText, MarthaDot} specifically for NLP. The second one was proposed in the form $\sqrt{\sigma}\, \rho\, \sqrt{\sigma}$ within the context of a quantum theory of Bayesian inference \cite{Leifer1, leifer2013towards, CoeckeSpekkens2012}. Clearly, as update mechanisms, each of these can be seen as a quantitative generalisation of (\ref{eq:updateP}):
\beq\label{eq:star1}
\sum_i x_i \Bigl(P_i\circ - \circ P_i\Bigr)\vspace{-2mm} 
\eeq
\beq \label{eq:star2}
\left(\sum_i x_i P_i\right)\circ - \circ \left(\sum_j x_j P_j\right) 
\eeq

Diagrammatically, using un-doubling, we can represent the spectral decomposition of the density matrix $\sigma$ as follows:
\[
\begin{tikzpicture}[scale=0.80]
	\begin{pgfonlayer}{nodelayer}
		\node [style=none] (0) at (0, 0) {};
		\node [style=langstate, boldedge] (1) at (0, 0.25) {\!\!\!\!\!$\sigma$\!\!\!\!\!};
		\node [style=none] (2) at (0, -1) {};
	\end{pgfonlayer}
	\begin{pgfonlayer}{edgelayer}
		\draw [style=boldedge] (2.center) to (0.center);
	\end{pgfonlayer}
\end{tikzpicture}\ \ \leadsto\ \ \ \begin{tikzpicture}[scale=0.80]
	\begin{pgfonlayer}{nodelayer}
		\node [style=langstate] (0) at (0.5, 2.25) {\!\!\!$x$\!\!\!};
		\node [style=none] (1) at (-0.75, 1) {};
		\node [style=none] (2) at (-0.75, 0.5) {};
		\node [style=none] (3) at (-0.75, -1.25) {};
		\node [style=none] (4) at (-0.75, 0) {};
		\node [style=map] (5) at (-0.25, 0.25) {$P$};
		\node [style=none] (6) at (0.5, 0.5) {};
		\node [style=none] (7) at (-1.75, 1) {};
		\node [style=none] (8) at (-1.75, -1.25) {};
	\end{pgfonlayer}
	\begin{pgfonlayer}{edgelayer}
		\draw [style=swap] (1.center) to (2.center);
		\draw [style=swap] (4.center) to (3.center);
		\draw [style=swap] (0) to (6.center);
		\draw [style=swap] (7.center) to (8.center);
		\draw [style=swap, bend left=90, looseness=1.50] (7.center) to (1.center);
	\end{pgfonlayer}
\end{tikzpicture} 
\]
and the fuzz and phaser seen as update mechanisms as in (\ref{eq:update}) then become:
\[
\input{mech3.tikz}\ \ = \ \ \begin{tikzpicture}[scale=0.80]
	\begin{pgfonlayer}{nodelayer}
		\node [style=none] (0) at (-2.25, 2.75) {};
		\node [style=none] (1) at (-2.25, -2) {};
		\node [style=none] (2) at (-2.25, -0.75) {};
		\node [style=none] (3) at (-2.25, -0.25) {};
		\node [style=none] (4) at (1, 2.75) {};
		\node [style=map] (5) at (-1.75, -0.5) {$P$};
		\node [style=none] (6) at (-1, 0) {};
		\node [style=none] (7) at (2.25, -0.25) {};
		\node [style=map] (8) at (1.5, -0.5) {$P$};
		\node [style=none] (9) at (1, -0.25) {};
		\node [style=none] (10) at (1, -0.75) {};
		\node [style=langstate] (11) at (2.25, 2) {\!\!\!$x$\!\!\!};
		\node [style=none] (12) at (1, -2) {};
		\node [style=white dot] (13) at (2.25, 1) {};
		\node [style=none] (14) at (-1, -0.25) {};
	\end{pgfonlayer}
	\begin{pgfonlayer}{edgelayer}
		\draw [style=swap] (0.center) to (3.center);
		\draw [style=swap] (2.center) to (1.center);
		\draw [style=swap] (4.center) to (9.center);
		\draw [style=swap] (10.center) to (12.center);
		\draw [style=swap] (13) to (7.center);
		\draw [style=swap, in=180, out=90, looseness=0.75] (6.center) to (13);
		\draw [style=swap] (11) to (13);
	\end{pgfonlayer}
\end{tikzpicture} 
\qquad\qquad\quad
\input{mech5.tikz}\ \ = \ \ \begin{tikzpicture}[scale=0.80]
	\begin{pgfonlayer}{nodelayer}
		\node [style=langstate] (0) at (-1, 1.5) {\!\!\!$x$\!\!\!};
		\node [style=none] (1) at (-2.25, 2.25) {};
		\node [style=none] (2) at (-2.25, -0.25) {};
		\node [style=none] (3) at (-2.25, -2) {};
		\node [style=none] (4) at (-2.25, -0.75) {};
		\node [style=map] (5) at (-1.75, -0.5) {$P$};
		\node [style=none] (6) at (-1, -0.25) {};
		\node [style=none] (7) at (1, -0.75) {};
		\node [style=langstate] (8) at (2.25, 1.5) {\!\!\!$x$\!\!\!};
		\node [style=none] (9) at (1, -0.25) {};
		\node [style=none] (10) at (2.25, -0.25) {};
		\node [style=map] (11) at (1.5, -0.5) {$P$};
		\node [style=none] (12) at (1, 2.25) {};
		\node [style=none] (13) at (1, -2) {};
	\end{pgfonlayer}
	\begin{pgfonlayer}{edgelayer}
		\draw [style=swap] (1.center) to (2.center);
		\draw [style=swap] (4.center) to (3.center);
		\draw [style=swap] (0) to (6.center);
		\draw [style=swap] (12.center) to (9.center);
		\draw [style=swap] (7.center) to (13.center);
		\draw [style=swap] (8) to (10.center);
	\end{pgfonlayer}
\end{tikzpicture} 
\]
where the state labeled $x$ is the vector $|x\rangle=\left( x_1\dots x_{n-1} \right)^T$, the box labeled $P$ is the linear map $\sum_i P_i \otimes \langle i|$, and the white dot is a spider (\ref{eq:white spider}).

\subsection{The fuzz}  

We call \raisebox{-1.5mm}{\epsfig{figure=FUZZZZZ.jpg,width=11pt}} the \em fuzz\em.
The coefficients $x_i$ in the spectral decomposition of $\sigma$ are interpreted to representing the lack of knowledge about which proposition $P_i$ is imposed on $\rho$. In other words, the fuzz  imposes a fuzzy proposition
on the density matrix, 
and returns a density matrix comprising the mixture of having imposed different propositions each yielding a term $P_i\circ\rho \circ P_i$. This reflects the manner in which we would update a quantum state if there is uncertainty on which projector is imposed on a system undergoing measurement.

\subsection{The phaser}\label{sec:phaser} 

We call \raisebox{-1.5mm}{\epsfig{figure=PHASERRR.jpg,width=11pt}} the \em phaser\em. To understand the effect of the phaser, 
we write the 2nd argument $\sigma$ in (\ref{pedal2})
in terms of rank-1 projectors $P_i :=|i\rangle\langle i|$ for an ONB $\left\{|i\rangle\right\}_i$, which can always be done by allowing some of the $x_i$'s to be the same. We have the following initial result relating the phaser to spiders: 

\begin{lemma} 
The phaser, when the 1st argument is pure, takes the form of a spider where the ONB in which the spider is expressed arises from diagonalisation of the 2nd argument. Setting $|x\rangle=\left( x_1\dots x_{n-1} \right)^T$, we have: 
\[
\Bigl(|\psi\rangle\langle\psi|\Bigr)
\,\raisebox{-1.5mm}{\epsfig{figure=PHASERRR.jpg,width=11pt}} 
\left( \sum_i x_i^2 |i\rangle\langle i|\right) 
= |\phi\rangle \langle\phi| 
\]
where:
\beq\label{eq:s50dd}
|\phi\rangle:=\ \   \begin{tikzpicture}[scale=0.80]
	\begin{pgfonlayer}{nodelayer}
		\node [style=none] (0) at (0, -1.5) {};
		\node [style=langstate] (1) at (1.25, 1) {\!\!\!$x$\!\!\!};
		\node [style=white dot] (2) at (0, -0.5) {};
		\node [style=langstate] (3) at (-1.25, 1) {\!\!\!$\psi$\!\!\!}; 
		\node [style=none] (4) at (-1.25, 0.75) {};
		\node [style=none] (5) at (1.25, 0.75) {};
	\end{pgfonlayer}
	\begin{pgfonlayer}{edgelayer}
		\draw [style=swap] (2) to (0.center);
		\draw [style=swap, in=165, out=-90, looseness=1.00] (4.center) to (2);
		\draw [style=swap, in=-90, out=15, looseness=1.00] (2) to (5.center);
	\end{pgfonlayer}
\end{tikzpicture}
\eeq
So in particular, the resulting density matrix is also pure.
\end{lemma} 
\begin{proof}
We have, using the fact the  $x_i$'s are real:
\beqa
 \Bigl(|\psi\rangle\langle\psi|\Bigr)  
\,\raisebox{-1.5mm}{\epsfig{figure=PHASERRR.jpg,width=11pt}}\,  
\left( \sum_i x_i^2 |i\rangle\langle i|\right) 
&=& 
\Biggl(\sum_i x_i |i\rangle\langle i|\Biggr)|\psi\rangle\langle\psi|\Biggl(\sum_j x_j|j\rangle\langle j|\Biggr)\\
&=& 
\Biggl(\sum_i\langle i|\psi\rangle x_i |i\rangle\Biggr)\Biggl(\sum_j\langle \psi |j\rangle x_j\langle j|\Biggr) \\
&=& 
\Biggl(\sum_i\langle i|\psi\rangle x_i |i\rangle\Biggr)\Biggl(\sum_j\overline{\langle j| \psi \rangle} x_j\langle j|\Biggr) \\
&=& 
\Biggl(\sum_i\psi_i x_i |i\rangle\Biggr)\Biggl(\sum_j\bar{\psi}_i \bar{x}_j\langle j|\Biggr) 
\eeqa
with $\psi_i:= \langle i|\psi\rangle$.  As the explicit form of the spider is: 
\[
\begin{tikzpicture}[scale=0.80]
	\begin{pgfonlayer}{nodelayer}
		\node [style=none] (0) at (0, -1.5) {};
		\node [style=white dot] (1) at (0, -0.5) {};
		\node [style=none] (2) at (-1.25, 0.75) {};
		\node [style=none] (3) at (1.25, 0.75) {};
	\end{pgfonlayer}
	\begin{pgfonlayer}{edgelayer}
		\draw [style=swap] (1) to (0.center);
		\draw [style=swap, in=165, out=-90, looseness=1.00] (2.center) to (1);
		\draw [style=swap, in=-90, out=15, looseness=1.00] (1) to (3.center);
	\end{pgfonlayer}
\end{tikzpicture}\ \ =\ \sum_i |i\rangle\langle ii|  
\]
we indeed have:
\[
\sum_i\psi_i x_i |i\rangle \ =\ |\phi\rangle \ = \ \ \begin{tikzpicture}[scale=0.80]
	\begin{pgfonlayer}{nodelayer}
		\node [style=none] (0) at (0, -1.5) {};
		\node [style=langstate] (1) at (1.25, 1) {\!\!\!$x$\!\!\!};
		\node [style=white dot] (2) at (0, -0.5) {};
		\node [style=langstate] (3) at (-1.25, 1) {\!\!\!$\psi$\!\!\!}; 
		\node [style=none] (4) at (-1.25, 0.75) {};
		\node [style=none] (5) at (1.25, 0.75) {};
	\end{pgfonlayer}
	\begin{pgfonlayer}{edgelayer}
		\draw [style=swap] (2) to (0.center);
		\draw [style=swap, in=165, out=-90, looseness=1.00] (4.center) to (2);
		\draw [style=swap, in=-90, out=15, looseness=1.00] (2) to (5.center);
	\end{pgfonlayer}
\end{tikzpicture} 
\]  
what completes the proof.
\end{proof}

From this it now  follows that the apparently obscure phaser, in particular due to the involvement of square root when presented as in \cite{Leifer1, leifer2013towards, CoeckeSpekkens2012}, canonically generalises to the spiders previously used in DisCoCat: 

\begin{theorem}\label{thm:phaserspider}
The action of the phaser on its first argument can be expressed in terms of spiders, explicitly, using the notations of (\ref{pedal2}), it takes the form:
\beq\label{eq:spideractionthm}
-\,\raisebox{-1.5mm}{\epsfig{figure=PHASERRR.jpg,width=11pt}}\, \sigma\ \ = \ \ \begin{tikzpicture}[scale=0.80]
	\begin{pgfonlayer}{nodelayer}
		\node [style=none] (0) at (0, -1.5) {};
		\node [style=none] (1) at (0, 2) {};
		\node [style=langstate, boldedge] (2) at (1.5, 1) {\!\!\!$x$\!\!\!};
		\node [style=white dot, boldedge] (3) at (0, -0.5) {};
		\node [style=none] (4) at (1.5, 0.75) {};
	\end{pgfonlayer}
	\begin{pgfonlayer}{edgelayer}
		\draw [style=boldedge] (3) to (0.center);
		\draw [style=boldedge] (1.center) to (3);
		\draw [style=boldedge, in=-90, out=15, looseness=1.00] (3) to (4.center);
	\end{pgfonlayer}
\end{tikzpicture} 
\eeq
\end{theorem}
\begin{proof}  
We have:
\[
\begin{tikzpicture}[scale=0.80]
	\begin{pgfonlayer}{nodelayer}
		\node [style=none] (0) at (0, -1.5) {};
		\node [style=langstate] (1) at (1.5, 1) {\!\!\!$x$\!\!\!};
		\node [style=white dot] (2) at (0, -0.5) {};
		\node [style=none] (3) at (0, 2) {};
		\node [style=none] (4) at (1.5, 0.75) {};
	\end{pgfonlayer}
	\begin{pgfonlayer}{edgelayer}
		\draw [style=swap] (2) to (0.center);
		\draw [style=swap, in=90, out=-90, looseness=0.75] (3.center) to (2);
		\draw [style=swap, in=-90, out=15, looseness=1.00] (2) to (4.center);
	\end{pgfonlayer}
\end{tikzpicture}\ \ = \ \Biggl(\sum_j |j\rangle\langle jj| \Biggr)\circ\Biggl(1\otimes \sum_i x_i |i\rangle \Biggr)\ = \ \sum_i x_i |i\rangle\langle i|
\]
which then yields the action of the phaser in the form (\ref{eq:star2}).
\end{proof}

So in conclusion, the phaser boils down to the spiders that we are already familiarly with in DisCoCat, hence now solidly justifying its consideration by us in the first place.  Moreover,  there is one important qualification that will overcome our objection voiced above against using spiders given that they yield commutativity.  Namely,  these spiders may be expressed in different ONBs which they inherit from the 2nd argument $\sigma$,
and if we update with nouns which diagonalise in different bases, then the corresponding spiders typically won't commute: 
\[
\begin{tikzpicture}[scale=0.80]
	\begin{pgfonlayer}{nodelayer}
		\node [style=none] (0) at (-1, 1.5) {};
		\node [style=white dot, boldedge] (1) at (0, 0.25) {\bf\tiny \,2\,};
		\node [style=none] (2) at (0, -2) {};
		\node [style=none] (3) at (0, 2) {};
		\node [style=none] (4) at (1, 0) {};
		\node [style=white dot, boldedge] (5) at (0, -1.25) {\bf\tiny \,1\,};
	\end{pgfonlayer}
	\begin{pgfonlayer}{edgelayer}
		\draw [style=boldedge, in=135, out=-90, looseness=0.75] (0) to (1);
		\draw [style=boldedge] (3.center) to (2.center);
		\draw [style=boldedge, in=45, out=-90, looseness=0.75] (4) to (5);
	\end{pgfonlayer}
\end{tikzpicture}\ \ \not=\ \ \begin{tikzpicture}[scale=0.80]
	\begin{pgfonlayer}{nodelayer}
		\node [style=none] (0) at (-1, 0) {};
		\node [style=white dot, boldedge] (1) at (0, -1.25) {\bf\tiny \,2\,};
		\node [style=none] (2) at (0, -2) {};
		\node [style=none] (3) at (0, 2) {};
		\node [style=none] (4) at (1, 1.5) {};
		\node [style=white dot, boldedge] (5) at (0, 0.25) {\bf\tiny \,1\,};
	\end{pgfonlayer}
	\begin{pgfonlayer}{edgelayer}
		\draw [style=boldedge, in=135, out=-90, looseness=0.75] (0) to (1);
		\draw [style=boldedge] (3.center) to (2.center);
		\draw [style=boldedge, in=45, out=-90, looseness=0.75] (4) to (5);
	\end{pgfonlayer}
\end{tikzpicture} 
\]
Hence, for the phaser, it is the properties with which the nouns are updated that control commutativity.
The special case in which they commute is then the counterpart to the intersective adjectives \cite{KampPartee1995} mentioned above. 

Finally, we justify the term `phaser'.
Recalling that the key feature of the fuzz and the phaser is that they produce a 
density matrix, we see that we can let the $x_i$ in the phaser be complex: 
\beq\label{eq:star2'} 
\Biggl(\sum_i x_i |i\rangle\langle i|\Biggr)\circ - \circ\Biggl(\sum_j \bar{x}_j|j\rangle\langle j|\Biggr)
\eeq
In that case, of course, the density matrix $\sigma:= \sum_i |x_i|^2 P_i$ does not fully specify (\ref{eq:star2'}), so rather than the density matrix, the data needed is the pair consisting of all $x_i$'s and $P_i$'s. Taking all $x_i$'s such that $|x_i|=1$ then the operation:
\ctikzfig{s50ddbis}
takes the form of the original \em phases \em of ZX-calculus \cite{CD1, CD2, CKbook}. 
All spiders are equipped with phases, and more abstractly, they can be defined as certain Frobenius algebras \cite{CPV}.
In more recent versions of the ZX-calculus, more general phases are also allowed \cite{ng2018completeness, DBLP:conf/rc/CoeckeW18}, as these exist for equally general abstract reasons, and this then brings us to the general case of the phaser. 

\subsection{Normalisation for fuzz and phaser}

We have the following no-go theorem for the fuzz:

\begin{prop}\label{prop:normalisation2}
If the operation (\ref{eq:star1}) sends normalised density matrices to normalised density matrices, then it must be equal to a (partial) decoherence operation:
\[
\sum_i \Bigl(P_i\circ - \circ P_i\Bigr) 
\]
which retains all diagonal elements and sets off-diagonal ones to zero.
\end{prop}
\begin{proof}
By trace preservation $\sum_i x_i (P_i\circ P_i) = \sum_i x_i P_i$ is the identity, so $x_i=1$.
\end{proof}

For the phaser we have an even stronger result:

\begin{prop}\label{prop:normalisation1}
If the operation (\ref{eq:star2'}) sends normalised density matrices to normalised density matrices, then for all $i$ we have $|x_i|=1$. Taking the $x_i$'s to be positive reals, only the identity remains.
\end{prop}
\begin{proof}
By trace preservation $(\sum_i x_i P_i)\circ (\sum_j \bar{x}_j P_j) = \sum_i x_i\bar{x}_i (P_i\circ P_i) = \sum_i |x_i|^2 P_i$ is the identity, so $|x_i|=1$.
\end{proof}

It immediately follows from Propositions~\ref{prop:normalisation2} and \ref{prop:normalisation1} that the operations (\ref{eq:star1}) and (\ref{eq:star2'}) only preserve normalisation for a single trivial action both of $\raisebox{-1.5mm}{\epsfig{figure=FUZZZZZ.jpg,width=11pt}}$ and $\raisebox{-1.5mm}{\epsfig{figure=PHASERRR.jpg,width=11pt}}$. Of course, this was already the case for single projectors $P_A$, which will only preserve normalisation for fixed-points, so this result shouldn't come as a surprise. Hence, just like in quantum theory, one needs to re-normalise after each update if one insists on density matrices to be normalised.

\subsection{Experimental evidence}\label{sec:exp}  

Both the fuzz and the phaser have recently been
numerically tested in their performance in modelling lexical entailment \cite{MarthaDot} (a.k.a.~hyponymy).
In \cite{MarthaDot} both fuzz and phaser are used to compose meanings of words in sentences, and it is explored how lexical entailment relationships propagate when doing so. The phaser performs particularly well, and seems to be very suitable when one considers more complex grammatical constructs. While these results were obtained within the context of DisCoCat, they also lift to the realm of DisCoCirc.

\section{Non-uniqueness and non-internalness}\label{sec:lotsofnon}

The above poses a dilemma; there are two candidates for meaning update mechanisms in DisCoCirc.
This seems to indicate that a DisCoCirc-formalism entirely based on updating, subject to that update process being unique, is not achievable.

Moreover, it is easy to check (and well-known for the phaser \cite{horsman2017can}) that both  $\raisebox{-1.5mm}{\epsfig{figure=FUZZZZZ.jpg,width=11pt}}$ and $\raisebox{-1.5mm}{\epsfig{figure=PHASERRR.jpg,width=11pt}}$ fail to have basic algebraic properties such as associativity, so treating them as algebraic connectives is not useful either. But that was never really our intension anyway, given that the formal framework where meanings in DisCoCat and DisCoCirc live is the theory of monoidal categories \cite{CatsII, SelingerSurvey}. In these categories, we both have states and processes which transform these states. In the case that states are vectors these process typically are linear maps, and in the case that states are density matrices these processes typically are CP-maps. However, neither the fuzz nor the phaser is a CP-map on the input $\rho\otimes\sigma$, which can clearly be seen from the double occurrence of projectors in their outputs. In other words, these update mechanisms are not \em internal \em to the meaning category. This means that there is no clear `mathematical arena' where they live. 

We will now move to a richer meaning category where
the fuzz and phaser will be unified in a single construction which will become internal to the meaning category, as well as having a diagrammatic representation. 

\section{Pedalboard with double mixing}\label{sec:unification}

In order to unify fuzz $\raisebox{-1.5mm}{\epsfig{figure=FUZZZZZ.jpg,width=11pt}}$ and phaser $\raisebox{-1.5mm}{\epsfig{figure=PHASERRR.jpg,width=11pt}}$\,, and also to make them internal to the meaning category, we use the \em double density matrices \em (DDMs) of \cite{Ashoush, Zwart2017}. This is a new mathematical entity initially introduced within the context of NLP for capturing both lexical entailment and ambiguity within one structure \cite{Ashoush}. On the other hand, they are a natural progression from the density matrices introduced by von Neumann for formulating quantum theory \cite{vNdensity}. The key feature of DMMs for us is that they have two distinct modes of mixedness, for which we have the following:

\begin{theorem}  
DMMs enable one to unify fuzz and phaser in a combined update mechanism, where fuzz and phaser correspond to the two modes of mixedness of DDMs. In order to do so, meanings of propositions are generalised to being DMMs, and the update dot is then entirely made up of wires only:
\ctikzfig{double8} 
\end{theorem} 

We now define DDMs, and continue with the proof of the theorem.
Firstly, it is shown in \cite{Zwart2017} that there are two natural classes of DDMs, namely those arising from double dilation, and those arising from double mixing, and here we need the latter. While mixing can be thought of as passing from vectors to weighted sums of doubled vectors as follows (using the un-doubling representation): 
\[
|\phi\rangle\quad \leadsto \quad\sum_i x_i |\phi_i\rangle|\bar\phi_i\rangle
\]
\em double mixing \em means repeating that process once more \cite{Ashoush}:
\[
\sum_i x_i |\phi_i\rangle|\bar\phi_i\rangle \quad \leadsto \quad
\sum_{ijk} y_k x_{ik} x_{jk}
|\phi_{ik}\rangle|\bar\phi_{ik}\rangle  |\phi_{jk}\rangle|\bar\phi_{jk}\rangle  
\]
Setting $|\omega_{ik} \rangle := y_{k}^{1/4} x_{ik}^{1/2} |\phi_{ik} \rangle$ this becomes:
\[
\sum_{ijk} |\omega_{ik}\rangle|\bar\omega_{ik}\rangle  |\omega_{jk}\rangle|\bar\omega_{jk}\rangle
\]
and in diagrammatic notation akin to that of un-doubled density matrices, we obtain the following generic form for double density matrices \cite{Zwart2017}: 
\[ 
\begin{tikzpicture}[scale=0.70]
	\begin{pgfonlayer}{nodelayer}
		\node [style=none] (0) at (-2, 0.75) {};
		\node [style=none] (1) at (-2, -1.25) {};
		\node [style=langstate] (2) at (-2, 1) {$\ \ \bar\omega\ $};
		\node [style=langstate] (3) at (2, 1) {$\ \ \omega\ $};
		\node [style=langstate] (4) at (6, 1) {$\ \ \bar\omega\ $};
		\node [style=langstate] (5) at (-6, 1) {$\ \ \omega\ $};
		\node [style=none] (6) at (6, -1.25) {};
		\node [style=none] (7) at (6, 0.75) {};
		\node [style=none] (8) at (2, -1.25) {};
		\node [style=none] (9) at (2, 0.75) {};
		\node [style=none] (10) at (-6, -1.25) {}; 
		\node [style=none] (11) at (-6, 0.75) {};
		\node [style=none] (12) at (7, 0.5) {};
		\node [style=none] (13) at (-7, 0.5) {};
		\node [style=none] (14) at (1, 0.5) {};
		\node [style=none] (15) at (-1, 0.5) {};
		\node [style=white dot] (16) at (0, -0.75) {};
		\node [style=none] (17) at (3, 0.5) {};
		\node [style=none] (18) at (5, 0.5) {};
		\node [style=none] (19) at (-5, 0.5) {};
		\node [style=none] (20) at (-3, 0.5) {};
	\end{pgfonlayer}
	\begin{pgfonlayer}{edgelayer}
		\draw [style=swap] (0.center) to (1.center);
		\draw [style=swap] (7.center) to (6.center);
		\draw [style=swap] (9.center) to (8.center);
		\draw [style=swap] (11.center) to (10.center);
		\draw [style=swap, in=180, out=-90, looseness=0.50] (13.center) to (16);
		\draw [style=swap, in=150, out=-90, looseness=1.00] (15.center) to (16);
		\draw [style=swap, in=-90, out=30, looseness=1.00] (16) to (14.center);
		\draw [style=swap, in=-90, out=0, looseness=0.50] (16) to (12.center);
		\draw [style=swap, bend right=90, looseness=1.00] (17.center) to (18.center);
		\draw [style=swap, bend right=90, looseness=1.00] (19.center) to (20.center);
	\end{pgfonlayer}
\end{tikzpicture} 
\] 
In order to relate DMMs to our discussion in Section \ref{sec:pedals}, we turn them into CP-maps in the un-doubled from of (\ref{eq:CPM4}), by bending up the inner wires:
 \beq\label{eq:double3} 
\begin{tikzpicture}[scale=0.70]
	\begin{pgfonlayer}{nodelayer}
		\node [style=langstate] (0) at (-2, 0.75) {$\ \ \bar\omega\ $};
		\node [style=langstate] (1) at (2, 0.75) {$\ \ \omega\ $};
		\node [style=langstate] (2) at (7, 0.75) {$\ \ \bar\omega\ $};
		\node [style=langstate] (3) at (-7, 0.75) {$\ \ \omega\ $};
		\node [style=none] (4) at (2, 0.25) {};
		\node [style=none] (5) at (2, 0.5) {};
		\node [style=none] (6) at (7, -2) {};
		\node [style=none] (7) at (7, 0.5) {};
		\node [style=none] (8) at (-7, -2) {};
		\node [style=none] (9) at (-7, 0.5) {};
		\node [style=none] (10) at (8, 0.25) {};
		\node [style=none] (11) at (-8, 0.25) {};
		\node [style=none] (12) at (1, 0.25) {};
		\node [style=none] (13) at (-1, 0.25) {};
		\node [style=white dot] (14) at (0, -1.25) {};
		\node [style=none] (15) at (4.5, 0.25) {};
		\node [style=none] (16) at (4.5, 2) {};
		\node [style=none] (17) at (3, 0.25) {};
		\node [style=none] (18) at (6, 0.25) {};
		\node [style=none] (19) at (-6, 0.25) {};
		\node [style=none] (20) at (-3, 0.25) {};
		\node [style=none] (21) at (-4.5, 0.25) {};
		\node [style=none] (22) at (-4.5, 2) {};
		\node [style=none] (23) at (-2, 0.25) {};
		\node [style=none] (24) at (-2, 0.5) {};
	\end{pgfonlayer}
	\begin{pgfonlayer}{edgelayer}
		\draw [style=swap] (5.center) to (4.center);
		\draw [style=swap] (7.center) to (6.center);
		\draw [style=swap] (9.center) to (8.center);
		\draw [style=swap, in=180, out=-90, looseness=0.50] (11.center) to (14);
		\draw [style=swap, in=150, out=-90, looseness=1.00] (13.center) to (14);
		\draw [style=swap, in=-90, out=30, looseness=1.00] (14) to (12.center);
		\draw [style=swap, in=-90, out=0, looseness=0.50] (14) to (10.center);
		\draw [style=swap] (16.center) to (15.center);
		\draw [style=swap, bend right=90, looseness=1.00] (4.center) to (15.center);
		\draw [style=swap, bend right=90, looseness=1.00] (17.center) to (18.center);
		\draw [style=swap, bend right=90, looseness=1.00] (19.center) to (20.center);
		\draw [style=swap] (24.center) to (23.center);
		\draw [style=swap] (22.center) to (21.center);
		\draw [style=swap, bend left=90, looseness=1.00] (23.center) to (21.center);
	\end{pgfonlayer}
\end{tikzpicture}\ \ = \ \ \begin{tikzpicture}[scale=0.70]
	\begin{pgfonlayer}{nodelayer}
		\node [style=none] (0) at (-2, 2.5) {};
		\node [style=none] (1) at (-2, 1.75) {};
		\node [style=langbox] (2) at (-2, 1.5) {$\ \ \bar\omega\ $};
		\node [style=langbox] (3) at (2, 1.5) {$\ \ \omega\ $};
		\node [style=none] (4) at (2, 1.75) {};
		\node [style=none] (5) at (2, 2.5) {};
		\node [style=none] (6) at (1, 1) {};
		\node [style=none] (7) at (-1, 1) {};
		\node [style=white dot] (8) at (0, 0) {};
		\node [style=none] (9) at (3, 1) {};
		\node [style=none] (10) at (3, -1) {};
		\node [style=none] (11) at (2, -1.75) {};
		\node [style=langbox] (12) at (-2, -1.5) {$\ \ \omega\ $};
		\node [style=none] (13) at (2, -2.5) {};
		\node [style=none] (14) at (-2, -1.75) {}; 
		\node [style=none] (15) at (-2, -2.5) {};
		\node [style=langbox] (16) at (2, -1.5) {$\ \ \bar\omega\ $};
		\node [style=none] (17) at (1, -1) {};
		\node [style=none] (18) at (-1, -1) {};
		\node [style=none] (19) at (-3, -1) {};
		\node [style=none] (20) at (-3, 1) {};
	\end{pgfonlayer}
	\begin{pgfonlayer}{edgelayer} 
		\draw [style=swap] (0.center) to (1.center);
		\draw [style=swap] (5.center) to (4.center);
		\draw [style=swap, in=150, out=-90, looseness=1.00] (7.center) to (8);
		\draw [style=swap, in=-90, out=30, looseness=1.00] (8) to (6.center);
		\draw [style=swap] (15.center) to (14.center);
		\draw [style=swap] (13.center) to (11.center);
		\draw [style=swap, in=90, out=-135, looseness=1.00] (8) to (18.center);
		\draw [style=swap, in=90, out=-45, looseness=1.00] (8) to (17.center);
		\draw [style=swap] (20.center) to (19.center);
		\draw [style=swap] (9.center) to (10.center);
	\end{pgfonlayer}
\end{tikzpicture}       
\eeq
where:  
\[
\begin{tikzpicture}[scale=0.70]
	\begin{pgfonlayer}{nodelayer}
		\node [style=none] (0) at (0, 1.25) {};
		\node [style=none] (1) at (0, 0.25) {};
		\node [style=langbox] (2) at (0, 0) {$\ \ \omega\ $};
		\node [style=none] (3) at (1, -0.25) {};
		\node [style=none] (4) at (1, -1.25) {};
		\node [style=none] (5) at (-1, -0.25) {};
		\node [style=none] (6) at (-1, -1.25) {};
	\end{pgfonlayer}
	\begin{pgfonlayer}{edgelayer} 
		\draw [style=swap] (0.center) to (1.center);
		\draw [style=swap] (3.center) to (4.center);
		\draw [style=swap] (5.center) to (6.center);
	\end{pgfonlayer}
\end{tikzpicture}\ \ := \ \ \begin{tikzpicture}[scale=0.70]
	\begin{pgfonlayer}{nodelayer}
		\node [style=langstate] (0) at (-1, 0.25) {$\ \ \omega\ $};
		\node [style=none] (1) at (-1, -0.25) {};
		\node [style=none] (2) at (-1, 0) {};
		\node [style=none] (3) at (0, -0.25) {};
		\node [style=none] (4) at (-2, -0.25) {};
		\node [style=none] (5) at (0, -1.5) {};
		\node [style=none] (6) at (-2, -1.5) {};
		\node [style=none] (7) at (1.5, 1.5) {};
		\node [style=none] (8) at (1.5, -0.25) {};
	\end{pgfonlayer}
	\begin{pgfonlayer}{edgelayer}
		\draw [style=swap] (2.center) to (1.center);
		\draw [style=swap] (3.center) to (5.center);
		\draw [style=swap] (4.center) to (6.center);
		\draw [style=swap] (7.center) to (8.center);
		\draw [style=swap, bend right=90, looseness=1.00] (1.center) to (8.center);
	\end{pgfonlayer}
\end{tikzpicture}\qquad\qquad\quad 
\begin{tikzpicture}[scale=0.70]
	\begin{pgfonlayer}{nodelayer}
		\node [style=none] (0) at (0, 1.25) {};
		\node [style=none] (1) at (0, 0.25) {};
		\node [style=langbox] (2) at (0, 0) {$\ \ \bar\omega\ $};
		\node [style=none] (3) at (1, -0.25) {};
		\node [style=none] (4) at (1, -1.25) {};
		\node [style=none] (5) at (-1, -0.25) {};
		\node [style=none] (6) at (-1, -1.25) {};
	\end{pgfonlayer}
	\begin{pgfonlayer}{edgelayer} 
		\draw [style=swap] (0.center) to (1.center);
		\draw [style=swap] (3.center) to (4.center);
		\draw [style=swap] (5.center) to (6.center);
	\end{pgfonlayer}
\end{tikzpicture}\ \ := \ \ \begin{tikzpicture}[scale=0.70]
	\begin{pgfonlayer}{nodelayer}
		\node [style=langstate] (0) at (-1, 0.25) {$\ \ \bar\omega\ $};
		\node [style=none] (1) at (-1, -0.25) {};
		\node [style=none] (2) at (-1, 0) {};
		\node [style=none] (3) at (0, -0.25) {};
		\node [style=none] (4) at (-2, -0.25) {};
		\node [style=none] (5) at (0, -1.5) {};
		\node [style=none] (6) at (-2, -1.5) {};
		\node [style=none] (7) at (-3.5, 1.5) {};
		\node [style=none] (8) at (-3.5, -0.25) {};
	\end{pgfonlayer}
	\begin{pgfonlayer}{edgelayer}
		\draw [style=swap] (2.center) to (1.center);
		\draw [style=swap] (3.center) to (5.center);
		\draw [style=swap] (4.center) to (6.center);
		\draw [style=swap] (7.center) to (8.center);
		\draw [style=swap, bend left=90, looseness=1.00] (1.center) to (8.center);
	\end{pgfonlayer}
\end{tikzpicture}
\]
as well as the horizontal reflections of these. In order to see that this is indeed an instance of (\ref{eq:CPM4}), using fusion (\ref{eq:fusion}) we rewrite the spider in the RHS of (\ref{eq:double3}) as follows:
\[
\begin{tikzpicture}[scale=0.70]
	\begin{pgfonlayer}{nodelayer}
		\node [style=none] (0) at (-4, -3) {};
		\node [style=none] (1) at (-4, -1.5) {};
		\node [style=langbox] (2) at (-4, -1.25) {$\ \ \omega\ $};
		\node [style=none] (3) at (-1, -0.75) {};
		\node [style=none] (4) at (-3, -0.75) {};
		\node [style=white dot] (5) at (-2, 0.25) {};
		\node [style=langbox] (6) at (-4, 1.75) {$\ \ \bar\omega\ $};
		\node [style=none] (7) at (-4, 2) {};
		\node [style=none] (8) at (-4, 3.25) {};
		\node [style=none] (9) at (-3, 1.25) {};
		\node [style=none] (10) at (-5, 1.25) {};
		\node [style=none] (11) at (-5, -0.75) {};
		\node [style=none] (12) at (3, 1.25) {};
		\node [style=none] (13) at (5, -0.75) {};
		\node [style=none] (14) at (1, -0.75) {};
		\node [style=langbox] (15) at (4, 1.75) {$\ \ \omega\ $};
		\node [style=none] (16) at (4, -1.5) {};
		\node [style=none] (17) at (4, 3.25) {};
		\node [style=white dot] (18) at (2, 0.25) {};
		\node [style=langbox] (19) at (4, -1.25) {$\ \ \bar\omega\ $};
		\node [style=none] (20) at (4, 2) {};
		\node [style=none] (21) at (4, -3) {};
		\node [style=none] (22) at (5, 1.25) {};
		\node [style=none] (23) at (3, -0.75) {}; 
		\node [style=none] (24) at (1, -2) {};
		\node [style=none] (25) at (-1, -2) {};
		\node [style=none] (26) at (6.25, -2) {};
		\node [style=none] (27) at (0.5, -2) {};
		\node [style=none] (28) at (6.25, 2.75) {};
		\node [style=none] (29) at (0.5, 2.75) {};
		\node [style=none] (30) at (-6.25, 2.75) {};
		\node [style=none] (31) at (-0.5, -2) {};
		\node [style=none] (32) at (-0.5, 2.75) {};
		\node [style=none] (33) at (-6.25, -2) {};
	\end{pgfonlayer}
	\begin{pgfonlayer}{edgelayer}
		\draw [style=swap] (0.center) to (1.center);
		\draw [style=swap, in=-150, out=90, looseness=1.00] (4.center) to (5);
		\draw [style=swap, in=90, out=-30, looseness=1.00] (5) to (3.center);
		\draw [style=swap] (8.center) to (7.center);
		\draw [style=swap, in=-90, out=135, looseness=1.00] (5) to (9.center);
		\draw [style=swap] (11.center) to (10.center);
		\draw [style=swap] (21.center) to (16.center);
		\draw [style=swap, in=-150, out=90, looseness=1.00] (14.center) to (18);
		\draw [style=swap, in=90, out=-30, looseness=1.00] (18) to (23.center);
		\draw [style=swap] (17.center) to (20.center);
		\draw [style=swap, in=-90, out=45, looseness=1.00] (18) to (12.center);
		\draw [style=swap] (13.center) to (22.center);
		\draw [style=swap] (25.center) to (3.center);
		\draw [style=swap] (24.center) to (14.center);
		\draw [style=swap, bend right=90, looseness=1.25] (25.center) to (24.center);
		\draw [style=dashed] (33.center) to (31.center);
		\draw [style=dashed] (31.center) to (32.center);
		\draw [style=dashed] (32.center) to (30.center);
		\draw [style=dashed] (33.center) to (30.center);
		\draw [style=dashed] (27.center) to (26.center);
		\draw [style=dashed] (26.center) to (28.center);
		\draw [style=dashed] (28.center) to (29.center);
		\draw [style=dashed] (27.center) to (29.center);
	\end{pgfonlayer}
\end{tikzpicture}
\]
These CP-maps take the concrete form:  
\beq\label{eq:ddsymbolic}
\sum_k \Biggl(\sum_i |\omega_{ik}\rangle\langle\omega_{ik}|\Biggr)\circ - \circ \Biggl(\sum_j |\omega_{jk}\rangle\langle\omega_{jk}|\Biggr) 
\eeq
where the $k$-summation sums corresponds to the spider and the other summations to the two connecting wires. Using the spectral decomposition of the density matrix $\sum_i|\omega_{ik}\rangle\langle\omega_{ik}|$ this can be rewritten as follows, where the $y_k$'s are arbitrary: 
\beqa
(\ref{eq:ddsymbolic})  
&=&
\sum_k \Biggl(\sum_i x_{ik}' P_{ik}\Biggr)\circ - \circ \Biggl(\sum_j x_{jk}' P_{jk}\Biggr)\\
&=&
\sum_k y_k \Biggl(\sum_i {x_{ik}'\over y_k} P_{ik}\Biggr)\circ - \circ \Biggl(\sum_j {x_{jk}'\over y_k} P_{jk}\Biggr)\\
&=&
\sum_k y_k \Biggl(\sum_i x_{ik} P_{ik}\Biggr)\circ - \circ \Biggl(\sum_j x_{jk} P_{jk}\Biggr)
\eeqa 
Now, this expression accommodates both the update mechanisms (\ref{eq:star1}) and (\ref{eq:star2}) as special cases, which are obtained by having either in the outer- or in the two inner-summations only a single index, and setting the corresponding scalar to $1$. Consequently, in this form, we can think of the doubled density matrices as a canonical generalisation of propositions that unifies fuzz and phaser.

By relying on idempotence of the projectors we obtain:
\[
(\ref{eq:ddsymbolic}) = \sum_k y_k \Biggl(\sum_i x_{ik} P_{ik}\circ P_{ik}\Biggr)\circ - \circ \Biggl(\sum_j x_{jk} P_{jk}\circ P_{jk}\Biggr)
\]
and we can now indicate the roles of fuzz and phaser diagrammatically, as follows:
\beq\label{eq:double6}  
\input{double6.tikz} 
\eeq
where the summations and corresponding scalars are represented by their respective pedals. Of course, this notation is somewhat abusive, as also the projectors are part of the fuzz and phaser. Also, while the phaser appears twice in this picture, there is only one, just like for a density matrix $|\psi\rangle\langle\psi|$ there are two occurrences of $|\psi\rangle$.

We now put (\ref{eq:double6}) in a form that exposes what the dot as in (\ref{eq:update}) is when taking meanings to be DMMs. Recalling the wire-bending we did in (\ref{eq:double3}) we have:
\[
\input{double6.tikz} \ \ = \ \ \input{double2.tikz}     
\]
so it indeed follows that the dot only contains plain wires, which completes the proof.

\begin{remark}
One question that may arise concerns the relationship of the decomposition of CP-maps (\ref{eq:double6}) and the Krauss decomposition of CP maps.
A key difference is that (\ref{eq:double6}) is a decomposition in terms of projections, involving two levels of sums, constituting it more refined or constrained than the more generic Krauss decomposition.
This then leads to interesting questions, for example, regarding uniqueness of the projectors and coefficients arising from the spectral decompositions in our pedal-board construction. Furthermore, these coefficients may be used in a quantitative manner, e.g.~for extracting entropies as in \cite{calco2015}.
\end{remark}  


\section{Meaning category and physical realisation}\label{sec:implementation}  

We can still take as meaning category density matrices with CP-maps as processes. We can indeed think of a DMM as a density matrix of the form (\ref{eq:CPonly2}):
\ctikzfig{double1copy}
The update dot is a CP-map of the form (\ref{eq:CPM4}): 
\ctikzfig{double1copycopy}
In this way we obtain a meaning category for which updating is entirely internal.  

In previous work we already indicated the potential 
benefits of the implementation of standard natural language processing tasks modelled in DisCoCat on a quantum computer \cite{WillC}. One major upshot is the exponential space gain one obtains by 
encoding large vectors on many-qubit states. Another is the availability of quantum algorithm that yield quadratic speedup for tasks such as classification and question-answering. Moreover, the passage from vectors to density matrices is natural for a quantum computer, as any qubit can also be in a mixed state.
Double density matrices are then implemented as mixed entangled states.  

\section{Some examples}\label{sec:examples}  

\paragraph{Example 1: Paint it black.} Following Theorem \ref{thm:phaserspider}, the phaser $\raisebox{-1.5mm}{\epsfig{figure=PHASERRR.jpg,width=11pt}}$ can be expressed  as a spider-action (\ref{eq:spideractionthm}). The aim of this example is to illustrate how in this form non-commutative updating arises.  For the sake of clarity we will only consider rank-1 projectors rather than proper phasers, but this suffices for indicating how non-commutativity of update arises.
The toy text for this example is:
\begin{center}
\texttt{Door turns red.\\ Door turns black.}
\end{center}
Diagrammatically we have:
\[  
\input{door_red_black1.tikz}
\]
One can think of {\tt turns} as an incarnation of {\tt is} with non-commutative capability.  Therefore it is of the form (\ref{eq:preupdate}), and reduces to (\ref{eq:update}), where the grey dots are spiders.  We take {\tt red} and {\tt black} to correspond to vectors $|r\rangle$ and $|b\rangle$, with induced density matrices $|r\rangle\langle r|$ and $|b\rangle\langle b|$. As they share the feature of both being colours, they are related (e.g.~according to a factual corpus) and won't be orthogonal:  
\ctikzfig{rbONBs} 
We can now build a ONB associated to {\tt red}, with $|r\rangle$ one of the basis vectors, and the other ones taken from $|r\rangle^\perp$, which results in red spiders representing {\tt red}, and then, taking $|x\rangle$ only to be non-zero for basis vector $|r\rangle$,  {\tt turns red} becomes: 
\[
\input{door_red_phaser.tikz} ~~=~~ \begin{tikzpicture}[scale=0.80]
	\begin{pgfonlayer}{nodelayer}
		\node [style=none] (0) at (-0.5, 2) {};
		\node [style=langstate, boldedge] (1) at (3.5, 1.25) {\tiny\!\!\!{ $x_r=1$, $x_{r^\perp}=0$ }\!\!\!};
		\node [style=red dot, boldedge] (2) at (-0.5, -0.5) {};
		\node [style=none] (3) at (-0.5, -0.75) {};
		\node [style=none] (4) at (-0.5, -2) {};
	\end{pgfonlayer}
	\begin{pgfonlayer}{edgelayer}
		\draw [style=boldedge, in=45, out=-90, looseness=1.00] (1) to (2);
		\draw [style=boldedge] (0.center) to (2);
		\draw [style=boldedge] (3.center) to (4);
	\end{pgfonlayer}
\end{tikzpicture}
\]
Similarly, {\tt turns black} becomes:
\[
\input{door_black_phaser.tikz} ~~=~~ \begin{tikzpicture}[scale=0.80]
	\begin{pgfonlayer}{nodelayer}
		\node [style=none] (0) at (-0.5, 2) {};
		\node [style=langstate, boldedge] (1) at (3.5, 1.25) {\tiny\!\!\!{ $x_b=1$, $x_{b^\perp}=0$ }\!\!\!};
		\node [style=black dot, boldedge] (2) at (-0.5, -0.25) {};
		\node [style=none] (3) at (-0.5, -0.5) {};
		\node [style=none] (4) at (-0.5, -2) {};
	\end{pgfonlayer}
	\begin{pgfonlayer}{edgelayer}
		\draw [style=boldedge, in=45, out=-90, looseness=1.00] (1) to (2);
		\draw [style=boldedge] (0.center) to (2);
		\draw [style=boldedge] (3.center) to (4);
	\end{pgfonlayer}
\end{tikzpicture} 
\]
Crucially, the red and black spiders won't commute since they are defined on different ONBs taken from $\{|r\rangle,|r\rangle^\perp \}$ and $\{|b\rangle,|b\rangle^\perp \}$ respectively. 

Concretely, we obtain two rank-1 projectors, $|r\rangle\langle r|$ and $|b\rangle\langle b|$ respectively, as desired.
Taking an initial  state for the door  only considering the door's colour: 
\[
\begin{tikzpicture}[scale=0.80]
	\begin{pgfonlayer}{nodelayer}
		\node [style=none] (0) at (-0.5, 0.25) {};
		\node [style=none] (1) at (-0.5, -1) {};
		\node [style=langstate, boldedge] (2) at (-0.5, 0.75) {\!\!\!\!{\tt door}\!\!\!\!};
	\end{pgfonlayer}
	\begin{pgfonlayer}{edgelayer}
		\draw [style=boldedge] (0.center) to (1);
	\end{pgfonlayer}
\end{tikzpicture}  \ \ := \ |d\rangle\langle d| 
\]
we obtain: 
\[
\input{door_red_black3.tikz} \ \ = \ \Bigl(|r\rangle\langle r|\Bigr) \circ \Bigl(|d\rangle\langle d|\Bigr)  \circ \Bigl(|r\rangle\langle r|\Bigr)\ = \ z\, |r\rangle\langle r|
\] 
for some non-zero $z\in \mathbb{R}^+$, so now the door is {\tt red},  and also:
\[
\input{door_red_black2.tikz}  \ \ = \ \Bigl(|b\rangle\langle b|\Bigr) \circ \Bigl(z\,|r\rangle\langle r|\Bigr)  \circ \Bigl(|b\rangle\langle b|\Bigr)\ = \ z'\, |b \rangle\langle b|
\]
so now the door is {\tt black}, just like Mick Jagger wanted it to be.



\paragraph{Example 2: Black fuzztones.}  The aim of this example is to demonstrate the operation of the fuzz  $\raisebox{-1.5mm}{\epsfig{figure=FUZZZZZ.jpg,width=11pt}}$ in modelling ambiguous adjectives. Above we treated  {\tt black} as pure, but in fact, it is ambiguous in that, for example, it may refer to a colour as well as to an art-gerne. 
This kind of ambiguity is accounted for by the fuzz.
Disambiguation may take place when applying the ambiguous adjective  to an appropriate noun, for example:
\begin{center}
\texttt{black poem}\\
\texttt{black door}
\end{center}
or not, when the noun is lexically ambiguous as well, for example:  
\begin{center}
\texttt{black metal}
\end{center}
which may be an art-genre, namely the music-genre, or the material:   
\[
\epsfig{figure=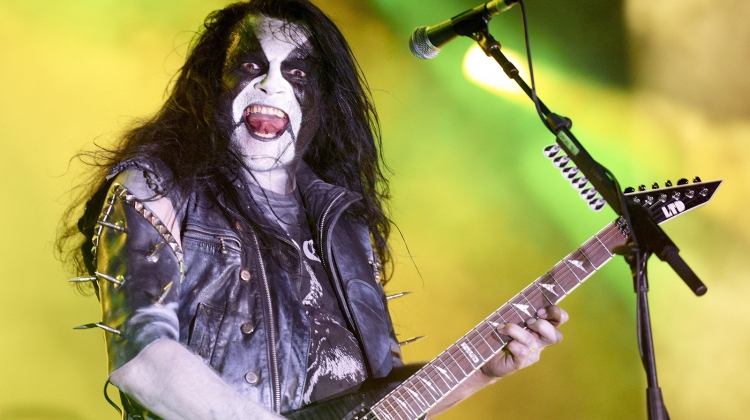,width=177pt}\qquad\qquad\epsfig{figure=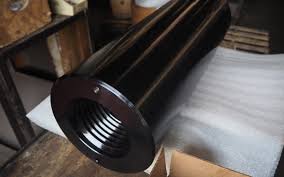,width=160pt}  
\]
This same ambiguity can propagate even further e.g.: 
\begin{center}
\texttt{black metal fan}
\end{center}
which clearly `demon'`strat'es the importance of the fuzz, as ambiguity is ubiquitous in natural language. The latter example  in fact also involves grammatical  and syntactical  ambiguity.
This level of ambiguity is beyond the scope of this paper and will be studied elsewhere.

So now, besides {\tt black}$_{col}$ as defined above for colour, there is another use of it, namely {\tt black}$_{gen}$ for genre: 
\[
\qquad\qquad\input{door_black_phaser-style.tikz} ~~=~~ \input{door_skull_spider.tikz}     
\]
and we can represent the overall meaning as the following fuzz:
\[
\input{black_fuzz.tikz} ~~=~~ y\ \input{door_black_phaser-colour.tikz}\ \ +\ \ y' \ \input{door_black_phaser-style.tikz}
\]
that is, concretely:
\[
- ~\raisebox{-1.5mm}{\epsfig{figure=FUZZZZZ.jpg,width=11pt}}  ~ \sigma_{\mbox{\tt black}} = 
y\, \Bigl( | b_{col} \rangle\langle b_{col} |  \circ - \circ  | b_{col}\rangle\langle b_{col} | \Bigr)
\, +\, 
y'\, \Bigl( | b_{gen}\rangle\langle b_{gen} |  \circ - \circ  | b_{gen}\rangle\langle b_{gen}| \Bigr)
\]
where the ambiguity is induced by the adjective: 
\[
\sigma_{\mbox{\tt black}} := | b_{col}\rangle\langle b_{col}| + | b_{gen}\rangle\langle b_{gen}|  
\]

Inputting $\rho_\mathrm{poem}$, $\rho_\mathrm{door}$, and $\rho_\mathrm{metal}$, we expect empirically 
(from a factual corpus) that  the following terms will be very small:
\[
\Bigl( | b_{col}\rangle\langle b_{col} |  \circ \rho_\mathrm{poem} \circ  | b_{col}\rangle\langle b_{col} | \Bigr) 
\approx 0 \approx
\Bigl( | b_{gen}\rangle\langle b_{gen} |  \circ \rho_\mathrm{door} \circ  | b_{gen}\rangle\langle b_{gen}| \Bigr)
\]
while these will all be significant:
\[
\Bigl( | b_{col}\rangle\langle b_{col} |  \circ \rho_\mathrm{door} \circ  | b_{col}\rangle\langle b_{col} | \Bigr) \gg 0
\qquad 
\Bigl( | b_{gen}\rangle\langle b_{gen} |  \circ \rho_\mathrm{poem} \circ  | b_{gen}\rangle\langle b_{gen}| \Bigr) \gg 0
\]
\[
\Bigl( | b_{col}\rangle\langle b_{col} |  \circ \rho_\mathrm{metal} \circ  | b_{col}\rangle\langle b_{col} | \Bigr) \gg 0
\qquad 
\Bigl( | b_{gen}\rangle\langle b_{gen} |  \circ \rho_\mathrm{metal} \circ  | b_{gen}\rangle\langle b_{gen}| \Bigr) \gg 0 
\]
That is, \texttt{poem} and \texttt{door} are  unambiguous nouns
that  disambiguate the ambiguous adjective \texttt{black}, while 
$\texttt{metal}$ is ambiguous before
and remains ambiguous after the application of the adjective \texttt{black} on it.   

\section{Outro} 

In this paper we proposed update mechanisms for DisCoCirc, in terms of fuzz and phaser,  which in the real world look something like this:
\begin{center}
\raisebox{-6mm}{\epsfig{figure=FUZZZZZ.jpg,width=30pt}}\ \ :=\ \ \raisebox{-9mm}{\epsfig{figure=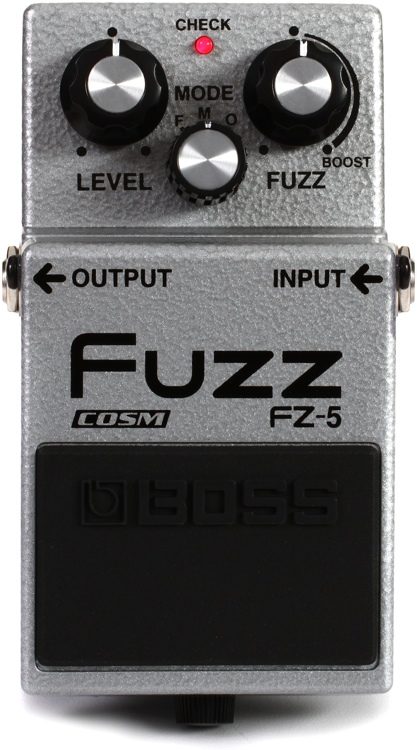,width=30pt}}
\qquad\qquad\qquad\qquad\qquad
\raisebox{-6mm}{\epsfig{figure=PHASERRR.jpg,width=30pt}}\ \ :=\ \ \raisebox{-9mm}{\epsfig{figure=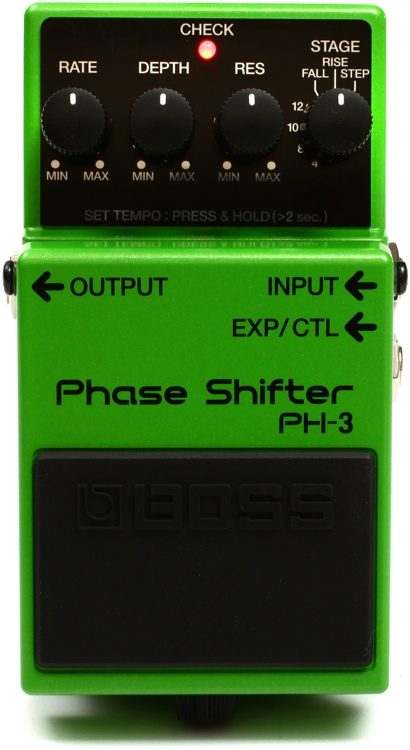,width=30pt}}
\end{center}
We unified them within a single diagrammatically elegant update mechanism by employing double density matrices.
In this way we upgraded the commutative spiders used in DisCoCat to non-commutative ones that respect the temporal order of the sentences within a text.  The commutative spiders consist a special case of the more general phaser. At the same time, the fuzz models lexical ambiguity.  


One might consider employing the double density matrix formalism to contribute to a theory of quantum Bayesian inference.
Vice versa, fully incorporating inference within the diagrammatic formalism of quantum theory would aid in successfully modeling tasks in natural language processing as well as cognition.

Furthermore, since double density matrices can be described by standard density matrices with post-selection,
the update formalism we have defined here can in principle be implemented in quantum hardware.
Therefore, our framework provides a source for small- and large-scale experiments in the novel field of quantum natural language processing.

Finally, one might wonder whether an `anatomy of completely positive maps' can be performed by means of double density matrices, 
potentially providing a compact framework in which to study quantum channels.

\bibliographystyle{plain}
\bibliography{main}

\begin{thebibliography}{10}

\bibitem{Ajdukiewicz}
K.~Ajdukiewicz.
\newblock Die syntaktische konnexit\"at.
\newblock {\em Studia Philosophica}, 1:1--27, 1937.

\bibitem{Ashoush}
D.~Ashoush and B.~Coecke.
\newblock {Dual Density Operators and Natural Language Meaning}.
\newblock {\em Electronic Proceedings in Theoretical Computer Science},
  221:1--10, 2016.
\newblock arXiv:1608.01401.

\bibitem{EsmaSC}
E.~Balkir, M.~Sadrzadeh, and B.~Coecke.
\newblock {\em Distributional Sentence Entailment Using Density Matrices},
  pages 1--22.
\newblock Springer International Publishing, Cham, 2016.

\bibitem{bankova2016graded}
D.~Bankova, B.~Coecke, M.~Lewis, and D.~Marsden.
\newblock Graded entailment for compositional distributional semantics.
\newblock {\em arXiv preprint arXiv:1601.04908}, 2018.
\newblock Accepted for publication.

\bibitem{BennettShor}
C.~H. Bennett and P.~W. Shor.
\newblock Quantum information theory.
\newblock {\em IEEE Trans. Inf. Theor.}, 44(6):2724--2742, September 2006.

\bibitem{BvN}
G.~Birkhoff and J.~von Neumann.
\newblock The logic of quantum mechanics.
\newblock {\em Annals of Mathematics}, 37:823--843, 1936.

\bibitem{blacoe2013quantum}
W.~Blacoe, E.~Kashefi, and M.~Lapata.
\newblock A quantum-theoretic approach to distributional semantics.
\newblock In {\em Procs.~of the 2013 Conf.~of the North American Chapter of the
  Association for Computational Linguistics: Human Language Technologies},
  pages 847--857, 2013.

\bibitem{ConcSpacI}
J.~Bolt, B.~Coecke, F.~Genovese, M.~Lewis, D.~Marsden, and R.~Piedeleu.
\newblock Interacting conceptual spaces {I}: Grammatical composition of
  concepts.
\newblock In M.~Kaipainen, A.~Hautam\"aki, P.~G\"ardenfors, and F.~Zenker,
  editors, {\em Concepts and their Applications}, Synthese Library, Studies in
  Epistemology, Logic, Methodology, and Philosophy of Science. Springer, 2018.
\newblock to appear.

\bibitem{CoeckeText}
B.~Coecke.
\newblock The mathematics of text structure, 2019.
\newblock arXiv:1904.03478.

\bibitem{CD1}
B.~Coecke and R.~Duncan.
\newblock Interacting quantum observables.
\newblock In {\em Proceedings of the 37th International Colloquium on Automata,
  Languages and Programming (ICALP)}, Lecture Notes in Computer Science, 2008.

\bibitem{CD2}
B.~Coecke and R.~Duncan.
\newblock Interacting quantum observables: categorical algebra and
  diagrammatics.
\newblock {\em New Journal of Physics}, 13:043016, 2011.
\newblock {arXiv:quant-ph/09064725}.

\bibitem{DBLP:conf/wollic/CoeckeGLM17}
B.~Coecke, F.~Genovese, M.~Lewis, and D.~Marsden.
\newblock Generalized relations in linguistics and cognition.
\newblock In Juliette Kennedy and Ruy J. G.~B. de~Queiroz, editors, {\em Logic,
  Language, Information, and Computation - 24th International Workshop, WoLLIC
  2017, London, UK, July 18-21, 2017, Proceedings}, volume 10388 of {\em
  Lecture Notes in Computer Science}, pages 256--270. Springer, 2017.

\bibitem{LambekvsLambek}
B.~Coecke, E.~Grefenstette, and M.~Sadrzadeh.
\newblock Lambek vs. {L}ambek: Functorial vector space semantics and string
  diagrams for {L}ambek calculus.
\newblock {\em Annals of Pure and Applied Logic}, 164:1079--1100, 2013.
\newblock arXiv:1302.0393.

\bibitem{CKpaperI}
B.~Coecke and A.~Kissinger.
\newblock Categorical quantum mechanics {I}: causal quantum processes.
\newblock In E.~Landry, editor, {\em Categories for the Working Philosopher}.
  Oxford University Press, 2016.
\newblock ar{X}iv:1510.05468.

\bibitem{CQMII}
B.~Coecke and A.~Kissinger.
\newblock Categorical quantum mechanics {II}: Classical-quantum interaction.
\newblock {\em International Journal of Quantum Information}, 14(04):1640020,
  2016.

\bibitem{CKbook}
B.~Coecke and A.~Kissinger.
\newblock {\em Picturing Quantum Processes. A First Course in Quantum Theory
  and Diagrammatic Reasoning}.
\newblock Cambridge University Press, 2017.

\bibitem{CLM}
B.~Coecke, M.~Lewis, and D.~Marsden.
\newblock Internal wiring of cartesian verbs and prepositions.
\newblock In M.~Lewis, B.~Coecke, J.~Hedges, D.~Kartsaklis, and D.~Marsden,
  editors, {\em {\rm Procs.~of the 2018 Workshop on} Compositional Approaches
  in Physics, NLP, and Social Sciences}, volume 283 of {\em Electronic
  Proceedings in Theoretical Computer Science}, pages 75--88, 2018.

\bibitem{CatsII}
B.~Coecke and {\'E}.~O. Paquette.
\newblock Categories for the practicing physicist.
\newblock In B.~Coecke, editor, {\em New Structures for Physics}, Lecture Notes
  in Physics, pages 167--271. Springer, 2011.
\newblock {a}rXiv:0905.3010.

\bibitem{CPaqPav}
B.~Coecke, {\'E}.~O. Paquette, and D.~Pavlovi{\'c}.
\newblock {Classical and quantum structuralism}.
\newblock In S.~Gay and I.~Mackie, editors, {\em Semantic Techniques in Quantum
  Computation}, pages 29--69. Cambridge University Press, 2010.
\newblock {a}rXiv:0904.1997.

\bibitem{CPV}
B.~Coecke, D.~Pavlovi{\'c}, and J.~Vicary.
\newblock A new description of orthogonal bases.
\newblock {\em Mathematical Structures in Computer Science, to appear},
  23:555--567, 2013.
\newblock {a}rXiv:quant-ph/0810.1037.

\bibitem{CSC}
B.~Coecke, M.~Sadrzadeh, and S.~Clark.
\newblock Mathematical foundations for a compositional distributional model of
  meaning.
\newblock In J.~van Benthem, M.~Moortgat, and W.~Buszkowski, editors, {\em A
  Festschrift for Jim Lambek}, volume~36 of {\em Linguistic Analysis}, pages
  345--384. 2010.
\newblock ar{x}iv:1003.4394.

\bibitem{CoeckeSpekkens2012}
B.~Coecke and R.~W. Spekkens.
\newblock Picturing classical and quantum bayesian inference.
\newblock {\em Synthese}, 186(3):651--696, 2012.

\bibitem{DBLP:conf/rc/CoeckeW18}
B.~Coecke and Q.~Wang.
\newblock {ZX}-rules for 2-qubit clifford+{T} quantum circuits.
\newblock In Jarkko Kari and Irek Ulidowski, editors, {\em Reversible
  Computation - 10th International Conference, {RC} 2018, Leicester, UK,
  September 12-14, 2018, Proceedings}, volume 11106 of {\em Lecture Notes in
  Computer Science}, pages 144--161. Springer, 2018.

\bibitem{GrefSadr}
E.~Grefenstette and M.~Sadrzadeh.
\newblock Experimental support for a categorical compositional distributional
  model of meaning.
\newblock In {\em The 2014 Conference on Empirical Methods on Natural Language
  Processing.}, pages 1394--1404, 2011.
\newblock ar{X}iv:1106.4058.

\bibitem{Grishin}
V.N. Grishin.
\newblock On a generalization of the {A}jdukiewicz-{L}ambek system.
\newblock In {\em Studies in nonclassical logics and formal systems}, pages
  315--334. Nauka, Moscow, 1983.

\bibitem{harris1954distributional}
Z.~S. Harris.
\newblock Distributional structure.
\newblock {\em Word}, 10(2-3):146--162, 1954.

\bibitem{horsman2017can}
D.~Horsman, C.~Heunen, M.~F. Pusey, J.~Barrett, and R.~W. Spekkens.
\newblock Can a quantum state over time resemble a quantum state at a single
  time?
\newblock {\em Proceedings of the Royal Society A: Mathematical, Physical and
  Engineering Sciences}, 473(2205):20170395, 2017.

\bibitem{KampPartee1995}
H.~Kamp and B.~Partee.
\newblock Prototype theory and compositionality.
\newblock {\em Cognition}, 57:129--191, 1995.

\bibitem{DimitriDPhil}
D.~Kartsaklis.
\newblock {\em Compositional Distributional Semantics with Compact Closed
  Categories and Frobenius Algebras}.
\newblock PhD thesis, University of Oxford, 2014.

\bibitem{KartSadr}
D.~Kartsaklis and M.~Sadrzadeh.
\newblock Prior disambiguation of word tensors for constructing sentence
  vectors.
\newblock In {\em The 2013 Conference on Empirical Methods on Natural Language
  Processing.}, pages 1590--1601. ACL, 2013.

\bibitem{KartsaklisSadrzadeh2014}
D.~Kartsaklis and M.~Sadrzadeh.
\newblock A study of entanglement in a categorical framework of natural
  language.
\newblock In {\em Proceedings of the 11th Workshop on Quantum Physics and Logic
  (QPL)}. Kyoto ‚Japan, 2014.

\bibitem{Lambek0}
J.~Lambek.
\newblock The mathematics of sentence structure.
\newblock {\em American Mathematics Monthly}, 65, 1958.

\bibitem{LambekBook}
J.~Lambek.
\newblock From word to sentence.
\newblock {\em Polimetrica, Milan}, 2008.

\bibitem{Leifer1}
M.~S. Leifer and D.~Poulin.
\newblock Quantum graphical models and belief propagation.
\newblock {\em Annals of Physics}, 323(8):1899--1946, 2008.

\bibitem{leifer2013towards}
M.~S. Leifer and R.~W. Spekkens.
\newblock Towards a formulation of quantum theory as a causally neutral theory
  of bayesian inference.
\newblock {\em Physical Review A}, 88(5):052130, 2013.

\bibitem{MarthaDot}
M.~Lewis.
\newblock Modelling hyponymy for discocat, 2019.
\newblock Proceedings of ACT 2019.

\bibitem{MarthaNeg}
M.~Lewis.
\newblock Towards negation in discocat, 2019.
\newblock Proceedings of SemSpace 2019.

\bibitem{marsden2017custom}
Dan Marsden and Fabrizio Genovese.
\newblock Custom hypergraph categories via generalized relations.
\newblock In {\em 7th Conference on Algebra and Coalgebra in Computer Science
  (CALCO 2017)}. Schloss Dagstuhl-Leibniz-Zentrum fuer Informatik, 2017.

\bibitem{ng2018completeness}
K.~F. Ng and Q.~Wang.
\newblock Completeness of the zx-calculus for pure qubit clifford+ t quantum
  mechanics.
\newblock {\em arXiv preprint arXiv:1801.07993}, 2018.

\bibitem{RobinMSc}
R.~Piedeleu.
\newblock Ambiguity in categorical models of meaning.
\newblock Master's thesis, University of Oxford, 2014.

\bibitem{calco2015}
Robin Piedeleu, Dimitri Kartsaklis, Bob Coecke, and Mehrnoosh Sadrzadeh.
\newblock Open system categorical quantum semantics in natural language
  processing.
\newblock In {\em 6th Conference on Algebra and Coalgebra in Computer Science
  (CALCO 2015)}. Schloss Dagstuhl-Leibniz-Zentrum fuer Informatik, 2015.

\bibitem{FrobMeanI}
M.~Sadrzadeh, S.~Clark, and B.~Coecke.
\newblock The {F}robenius anatomy of word meanings {I}: subject and object
  relative pronouns.
\newblock {\em Journal of Logic and Computation}, 23:1293--1317, 2013.
\newblock ar{X}iv:1404.5278.

\bibitem{FrobMeanII}
M.~Sadrzadeh, S.~Clark, and B.~Coecke.
\newblock The {F}robenius anatomy of word meanings {II}: possessive relative
  pronouns.
\newblock {\em Journal of Logic and Computation}, 26:785--815, 2016.
\newblock arXiv:1406.4690.

\bibitem{SelingerCPM}
P.~Selinger.
\newblock Dagger compact closed categories and completely positive maps.
\newblock {\em Electronic Notes in Theoretical Computer Science}, 170:139--163,
  2007.

\bibitem{SelingerSurvey}
P.~Selinger.
\newblock A survey of graphical languages for monoidal categories.
\newblock In B.~Coecke, editor, {\em New Structures for Physics}, Lecture Notes
  in Physics, pages 275--337. Springer-Verlag, 2011.
\newblock {a}rXiv:0908.3347.

\bibitem{vNdensity}
J.~von Neumann.
\newblock Wahrscheinlichkeitstheoretischer aufbau der quantenmechanik.
\newblock {\em Nachrichten von der Gesellschaft der Wissenschaften zu
  G{\"o}ttingen, Mathematisch-Physikalische Klasse}, 1:245--272, 1927.

\bibitem{vN}
J.~von Neumann.
\newblock {\em Mathematische grundlagen der quantenmechanik}.
\newblock Springer-Verlag, 1932.
\newblock Translation, {\it Mathematical foundations of quantum mechanics},
  Princeton University Press, 1955.

\bibitem{Widdows}
D.~Widdows.
\newblock Orthogonal negation in vector spaces for modelling word-meanings and
  document retrieval.
\newblock In {\em 41st Annual Meeting of the Association for Computational
  Linguistics}, Japan, 2003.

\bibitem{widdows2003word}
D.~Widdows and S.~Peters.
\newblock Word vectors and quantum logic: Experiments with negation and
  disjunction.
\newblock {\em Mathematics of language}, 8(141-154), 2003.

\bibitem{WillC}
W.~Zeng and B.~Coecke.
\newblock Quantum algorithms for compositional natural language processing.
\newblock {\em Electronic Proceedings in Theoretical Computer Science}, 221,
  2016.
\newblock arXiv:1608.01406.

\bibitem{Zwart2017}
M.~Zwart and B.~Coecke.
\newblock {Double Dilation $\neq$ Double Mixing}.
\newblock 2017.
\newblock arXiv:1704.02309.

\end{thebibliography}

\end{document}